\documentclass[11pt,letterpaper]{article}

\usepackage[margin=1in]{geometry}
\usepackage{amsfonts,amsmath,dsfont,mathtools}
\usepackage{amssymb} %ACM
\usepackage[hypertexnames=false]{hyperref} %ACM
\usepackage{amsthm,comment}
\usepackage[inline, shortlabels]{enumitem}
\usepackage{color}
\usepackage[ruled]{algorithm2e}
\usepackage{xcolor}
\usepackage{xspace}
\usepackage{graphicx}

\hypersetup{ 
    colorlinks=true,
    citecolor=red,
    linkcolor=red,
    urlcolor=red
}
\graphicspath{{./figures/}}

\usepackage%[textsize=tiny]
{todonotes}
\usepackage[noabbrev,capitalise,nameinlink]{cleveref}
\crefname{enumi}{}{}
\crefname{equation}{}{}
\crefname{claim}{Claim}{Claims}

\newtheorem{theorem}{Theorem} %[section]
\newtheorem{lemma}[theorem]{Lemma}

\newtheorem{claim}[theorem]{Claim}
\newtheorem{definition}[theorem]{Definition}
\newtheorem*{definition*}{Definition}
\theoremstyle{remark}

\newtheorem{fact}{Fact}

\newcommand{\Prob}[1]{\mathbf{Pr}\left[#1\right]}
\newcommand{\Expc}[1]{\mathbf{E}\left[#1\right]}

\newcommand{\pa}[1]{\left( #1 \right)}

\newcommand{\minop}{\mathcal{O}}

\newcommand{\red}{\textsc{Red}}

\newcommand{\orange}{\textsc{Orange}}

\newcommand{\yellow}{\textsc{Yellow}}
\newcommand{\blue}{\textsc{Blue}}

\newcommand{\green}{\textsc{Green}}

\newcommand{\minority}{$k$-\textsc{minority}}
\newcommand{\majority}{$k$-\textsc{majority}}

\newcommand{\gossip}{\mbox{GOSSIP}}
\newcommand{\pull}{\mbox{PULL}}
\newcommand{\local}{\mbox{LOCAL}}
\newcommand{\minor}{m} 				% Number of agents with minority opinions
\newcommand{\eventA}{\mathcal{A}} 	% Wrong majority observed
\newcommand{\eventB}{\mathcal{B}} 	% No minority opinion observed

\renewcommand{\leq}{\leqslant}
\renewcommand{\le}{\leqslant}
\renewcommand{\geq}{\geqslant}
\renewcommand{\ge}{\geqslant}

\title{The Minority Dynamics and the Power of Synchronicity}

\author{
Luca Becchetti\footnote{\textit{Sapienza} University of Rome, Rome, Italy. becchetti@diag.uniroma1.it. Supported by the ERC Advanced Grant 788893 AMDROMA, PNRR MUR project PE0000013-FAIR”.}
\and
Andrea Clementi\footnote{\textit{Tor Vergata} University of Rome, Rome, Italy. clementi/pasquale@mat.uniroma2.it. Partially supported by Spoke1  ``FutureHPC \& BigData'' of the \textit{Italian Research Center on High-Performance Computing, Big Data and Quantum Computing (ICSC)} funded by \textit{MUR Missione 4 Componente 2 Investimento 1.4: Potenziamento strutture di ricerca e creazione di ``campioni nazionali'' di R\&S (M4C2-19) - Next Generation EU (NGEU).} }
\and
Francesco Pasquale$^{\dagger}$
\and
Luca Trevisan\footnote{Bocconi University, Milan, Italy. l.trevisan/isabella.ziccardi@unibocconi.it. This project has received funding from the European Research Council (ERC) under the European Union’s Horizon 2020 research and innovation programme (grant agreement No. 834861).}
\and
Robin Vacus\footnote{CNRS, IRIF, Paris, France. rvacus@irif.fr. } 
\and
Isabella Ziccardi$^{\ddagger}$
}

\date{}

\newcommand{\highlight}[1]{#1}

\begin{document}
\maketitle

\begin{abstract}

We study the minority-opinion dynamics over a fully-connected network of $n$ nodes with binary opinions.
Upon activation, a node receives a sample of opinions from a limited number of neighbors chosen uniformly at random.
Each activated node then adopts the opinion that is least common within the received sample.

Unlike all other known consensus dynamics, we prove that this elementary protocol behaves in dramatically different ways, depending on whether activations occur sequentially or in parallel.
Specifically, we show that its expected consensus time is exponential in $n$ under asynchronous models, such as  asynchronous $\gossip$.
On the other hand, despite its chaotic nature, we show that it converges within $O(\log^2 n)$ rounds with high probability under synchronous models, such as synchronous $\gossip$.

Finally, our results shed light on the bit-dissemination problem, that was previously introduced to model the spread of information in biological scenarios.
Specifically, our analysis implies that the minority-opinion dynamics is the first \textit{stateless} solution to this problem, in the parallel {passive-communication} setting, achieving convergence within a polylogarithmic number of rounds.
This, together with a known lower bound for sequential stateless dynamics, implies a parallel-vs-sequential gap for this problem that is nearly quadratic in the number $n$ of nodes. This is in contrast to all known results for problems in this area, which exhibit a linear gap between the parallel and the sequential setting.

\bigskip

\end{abstract}

\section {Introduction}

\subsection{Dynamics and Consensus}

In distributed computing, the term {\em dynamics} is used to refer to  distributed processes in which each node of a network updates its state on the basis of a simple update rule applied to messages received in the last round of communication; furthermore, the network is  {\em anonymous}, meaning that nodes do not have distinguished names and incoming messages lack any identification of the sender.

The study of dynamics, of their global properties and of their applications  is an active research topic that touches several scientific areas \cite{AAE08,BCNPT16,BCN20,DGMSS11,FHK14,MNT14}. Among other applications, dynamics are suitable to model the restrictions on computation and communication of IoT protocols, and to study the way information spreads and groups self-organize in biological models and in models of animal behavior.

In this paper we consider {\em opinion dynamics}. These are dynamics in which every node has an opinion, which is an element of a finite set, and possibly, additional state information. Over time, these opinions are affected by the information exchanged by nodes. Typically, one is interested in whether all nodes eventually reach a \textit{consensus}  configuration in which they all share the same opinion, and in how quickly this happens. 

The \textit{consensus} problem is a fundamental task in distributed computing \cite{Aspnes12,DeGroot74,DKS10,R83} and a problem of considerable interest in the study of dynamics. We provide the definition of a basic version of the problem.

\begin{definition}[Consensus Problem] We say that a dynamics in which every node, in every round, holds an opinion solves the {\em consensus problem} if
\begin{itemize}
    \item For every initial configuration of opinions, the dynamics  reaches with probability 1 a configuration in which all nodes have the same opinion, and such opinion was held by at least one node in the initial configuration;
    \item Once all nodes have the same opinion, the configuration of opinions does not change in subsequent rounds.
\end{itemize}
The number of rounds until all nodes share the same opinion for the first time is called  {\em convergence time}. For every initial configuration, convergence time is a random variable over the randomness of the dynamics.
\end{definition}

Since dynamics are meant to capture agents with very limited computational and communication abilities, several models have been defined and analyzed that place restrictions on how communication takes place (e.g. $\local$, various versions of $\gossip$,  POPULATION PROTOCOL, etc. \cite{aspnes_introduction_2009,peleg2000distributed,shah2009gossip}).
All dynamics we consider in this paper operate in a standard model in which  communication in each round is restricted, so that each active node (that is, each node that updates its opinion in that round) is only able to receive messages from $k$ randomly chosen neighbors (where $k$ is a parameter of the model). This is called $k$-uniform $\gossip$~$\pull$ in the literature \cite{boyd2006randomized,clementi2020consensus,shah2009gossip} and we will refer to it as $k$-$\pull$ in the remainder.

Moreover, we will always assume that the underlying communication network is the complete graph on $n$ nodes. We will be interested in two well-studied variants of this model \cite{boyd2006randomized}, which differ for the schedule with which nodes are activated: in the {\em asynchronous sequential} $k$-$\pull$, in each round exactly one node becomes active, while in the {\em synchronous parallel} $k$-$\pull$, all nodes are active and their updates are simultaneous in each round.

In both settings described above, a solution to the \textit{consensus} problem is given by the \textsc{voter model} \cite{A13,L12,peleg2000distributed}, which is the dynamics in which each active node receives the opinion of one random neighbor and then adopts that opinion as its own (thus the voter model operates as a $1$-$\pull$ dynamics).

It is known that, for every initial Boolean configuration of opinions, the \textsc{voter model}'s expected convergence time is $\Theta(n^2)$ in the asynchronous sequential setting \cite{A13,L12} and  $\Theta (n)$  in the synchronous parallel setting  \cite{HP01}.

In the \majority\ dynamics, each active node receives the Boolean opinions of $k$ random neighbors ($k$ odd) and updates its opinion to the majority opinion among those $k$. This is a $k$-$\pull$ dynamics, and the \textsc{voter model} can be seen as implementing $1$-\textsc{majority}. It is easy to see that \majority\ also provides a solution to the \textit{consensus} problem for every odd $k$, and it is known  \cite{DGMSS11} that, for odd $k\geq 3$,  in the synchronous parallel model, 
for every initial configuration, convergence time is at most $O(\log n)$ both in expectation and in high probability\footnote{We say that an event holds {\em with high probability} if it holds with probability $1-n^{\Omega(1)}$.}. It is folklore that for odd $k\geq 3$, the convergence time in the asynchronous sequential model is $O(n\log n)$. The \majority\ dynamics has other desirable properties, such as being fault-tolerant and such that the consensus opinion is likely to be the majority, or close to the majority, of the original opinions. Moreover, bounds on the convergence time of   $3$-\textsc{majority} have been derived   for  the case of  non-binary opinions in \cite{BCNPT16,BCN20,BCEKMN17,DGMSS11,GL17}.

We note that in both above examples above, there is a $\Theta (n)$ gap between parallel and sequential convergence times. Moreover, a  $\tilde \Theta(n)$ gap on the convergence time holds between sequential and parallel communication models for   other forms of \textit{consensus} as well \cite{AAE08,BCN20,ghaffari2016polylogarithmic}.

 This is natural, because we are applying the same update rule, but $n$ times simultaneously in each round in the parallel case, and once per round in the sequential case. Indeed, in all the dynamics for which we have tight bounds both in the sequential and parallel models we see a $\tilde \Theta(n)$ gap between them. Because of synchronicity of the updates however, one round in the parallel model might potentially have different effects than $n$ rounds in the sequential model, and we will return to this point in our main results.

\subsection{The bit-dissemination problem}

More recently, inspired by the study of distributed biological systems,  a variant of \textit{consensus}, called the \textit{bit-dissemination} problem has been considered in \cite{boczkowski_minimizing_2017,bastide_self-stabilizing_2021,DBLP:conf/podc/KormanV22,BCKPT23}.

\begin{definition} [Bit-Dissemination Problem] In the \emph{bit-dissemination} problem,  every node, at every round, holds an opinion. One of the nodes, called the {\em source},  holds an opinion that it knows to be correct, and never changes it. The other nodes update their opinions according to the update rule and do not know the identity of the source. We say that a dynamics solves the {\em bit-dissemination} problem if:
\begin{itemize}
    \item For every initial configuration, with  probability 1,
    the dynamics eventually reaches a configuration in which all nodes share the same opinion as the source;
    \item Once all nodes have the same opinion as the source, the configuration of opinions does not change in subsequent rounds.
\end{itemize}
The number of rounds until all nodes have the same opinion as the source for the first time  is called the {\em convergence time}.
\end{definition}

This information-dissemination problem finds its main motivations in modeling communication processes that take place in biological systems \cite{aspnes_introduction_2009,vicsek_collective_2012}, where the correct information may include, e.g., knowledge about a preferred migration
route~\cite{johnstone_information_2002,lindauer_communication_1957}, the location of a food source~\cite{couzin_uninformed_2011}, or the need to recruit agents for a particular task~\cite{razin_desert_2013}. 

In these applications, it is well motivated to restrict to dynamics that use {\em passive-communication}, meaning that the only kind of message that a node can receive from a neighbor is the neighbor's current opinion. Note that the \textsc{voter model} and the \majority\ dynamics that we described above are examples of passive-communication opinion dynamics.

Another desirable property in these \highlight{applications} is for dynamics to be {\em self-stabilizing}. This means that not only does the dynamics converge to the desired stable configuration for every initial configuration of the opinions, but it also converges regardless of the initial contents of any additional state that the update rule relies on. The reason that this is called a self-stabilizing property is that if, at some point, an adversary corrupts the internal state of some nodes, convergence is still assured (because we could take the round in which this corruption happens as the first round).

We call an opinion dynamics {\em stateless} if the update rule by which an active node updates its opinion depends only on the messages received in that round, and if nodes do not keep any additional state information except for their opinions.  Note that a stateless dynamics is always self-stabilizing because there is nothing for an adversary to corrupt. The \textsc{voter model} and \majority\ dynamics are examples of stateless dynamics.

In \cite{DBLP:conf/podc/KormanV22}, Korman and Vacus  discuss the difficulty of developing passive-communication and self-stabilizing dynamics for the \textit{bit-dissemination} problem.
To see why passive-communication is difficult to achieve for the \textit{bit-dissemination} problem, consider that if, at some round, the opinion of the source is in the minority, then we want to use an update rule that increases the number of nodes with the minority opinion. On the contrary, if the opinion of the source is in the majority, then we want an update rule that decreases the number of nodes with the minority opinion. These two scenarios, however, are indistinguishable to nodes that are only allowed to see other nodes' opinions and that do not know the identity of the source.

Indeed, previous protocols for this problem exploit non-passive forms of communication and/or do not achieve self-stabilizing \textit{bit-dissemination}~\cite{ben-or_fast_2008,boczkowski_minimizing_2017,bastide_self-stabilizing_2021}.
Korman and Vacus \cite{DBLP:conf/podc/KormanV22} develop and analyze a protocol for \textit{bit-dissemination} in the synchronous parallel $k$-$\pull$ with passive communication   in which every node, at every round, sees the opinion of $k=\Theta(\log n)$ random nodes (that is, communication follows the $k$-$\pull$ model), and they show that their protocol converges with high probability in  $O(\log^{5/2}n)$ rounds. 
Their protocol is not stateless, as it involves a trend-following update rule: roughly speaking, each node adopts the opinion whose number of occurrences in the messages of the current round has increased compared to the previous round, breaking ties arbitrarily. This requires each node to store the number of opinions of either type seen in the previous round, which take $\log k = \log\log n + O(1)$ bits of storage. Their protocol  is self-stabilizing, and so it converges even if, in the first round,  nodes start with adversarially chosen ``false memories'' of a previous round.
Their work leaves open whether convergence  in ${\rm poly}\log n$ rounds for the \textit{bit-dissemination} problem can be achieved, in the synchronous parallel $k$-$\pull$, by a passive-communication {\em stateless} dynamics, as is possible for the \textit{consensus} problem.

Becchetti et al. \cite{BCKPT23} show that in the asynchronous sequential $k$-$\pull$, there exists an $\Omega(n^2)$     lower bound for the expected convergence time of \emph{any} passive-communication stateless dynamics for the \textit{bit-dissemination} problem. This result holds under the even stronger model that, in every round, the active node is given full access to the opinions of {\em all} other nodes.

Given that there is typically a $\tilde \Theta (n)$ gap between the convergence time of stateless passive-communication dynamics in the parallel versus sequential model, it would seem natural to conjecture a $\tilde \Omega (n)$  lower bound for the convergence time of  \textit{bit-dissemination} in the synchronous parallel passive-communication $k$-$\pull$ for stateless dynamics, and that the memory of the past round used by Korman and Vacus would be essential to achieve ${\rm poly}\log n$  convergence time. A bit surprisingly however, we show this is not the case, as we discuss below.

\subsection {Our results}

We study the {\em minority} update rule,
%\footnote{Amos Korman came up with the idea of using this rule to tackle the bit-dessimination problem.}
that defines a stateless dynamics in the passive-communication setting in the $k$-$\pull$ communication model.

\begin{definition}[\minority\ Dynamics\footnote{\highlight{The authors acknowledge Amos Korman for proposing this protocol to solve the bit-dissemination problem.}}]
In each round, each active node, upon seeing the opinions of $k$ random neighbors, updates its opinion to the minority opinion among those held by the $k$ neighbors. Ties are broken randomly.
\end{definition}

We will only consider binary opinions and, for simplicity, odd $k$, so that ties never occur. According to this rule, if a node sees unanimity among its $k$ sampled neighbors, it adopts those neighbors' unanimous opinion. Otherwise, it adopts whatever opinion is held by the fewer neighbors.

In the \textit{consensus} problem, for any initial configuration, this dynamics eventually, with probability 1, reaches a configuration in which all the nodes have the same opinion, and, in the \textit{bit-dissemination} problem, it reaches with probability 1 the configuration in which all nodes agree with the source. This is because, in the Markov chain that describes the change of the configuration of opinions over rounds, the only absorbing states are those in which all nodes agree, and such configurations are reachable from any initial configuration.

Concerning the speed with which such convergence happens, consider the sequential setting first.
    In the asynchronous sequential model, if we start from a balanced configuration in which each opinion is held by 50\% of the nodes, it is very difficult for the \minority~dynamics to make progress, because as soon as an opinion gains more followers than the other, in subsequent rounds there will be a bias toward adopting the minority opinion and the configuration will drift back towards the initial 50\%-50\% configuration. Indeed, we can prove an exponential lower bound for the round complexity of \minority.  

\begin{theorem} \label{thm:lb_sequential}
    There exist initial configurations from which the expected convergence time of \minority~in the asynchronous sequential $k$-$\pull$ is $2^{\Omega(k n)}$.
\end{theorem}

In the synchronous parallel communication model, the \minority\ dynamics exhibits chaotic behavior. If, for example, the dynamics starts from a balanced configuration, in the next round we may expect a slight imbalance due to random noise, let's say slightly more 0s than 1s. In the subsequent round, this will cause a drift in the opposite direction, and we will see more 1s than 0s, with an increased imbalance. The majority value will keep alternating, and the imbalance   will keep increasing, until a round in which the imbalance becomes such that the minority opinion is held by $\ll n/k$ nodes. At that point the number of nodes with the minority opinion will grow by about a factor of $k$ each time, and then the majority value will start alternating again, and so on.

The main technical contribution of this paper is that, despite this chaotic behavior, when $k =\Omega( \sqrt{n\log n})$, we are able to analyze the behavior of the minority dynamics, by dividing the set of possible configurations into a finite number of ranges, and by understanding how likely it is for the configuration to move across these ranges. 

Our analysis shows that \minority\ converges in polylogarithmic time in the parallel setting, both when all nodes follow the protocol (\textit{consensus} problem) and when one node is a source that never changes its opinion (\textit{bit-dissemination} problem).

\begin{theorem} \label{thm:main}
If $ 185 \sqrt{n \log n} \leq k \leq  \frac{n}{2}$, then, in the synchronous parallel $k$-$\pull$,  \minority\ solves   the \emph{consensus}  and its convergence time is $O(\log n)$ in expectation and $O(\log^2 n)$ rounds w.h.p. 
\end{theorem}

This also provides an efficient (i.e.~in a polylogarithmic number of rounds) solution to the \textit{bit-dissemination} problem via a stateless passive-communication dynamics, answering a question left open by the work of Korman and Vacus \cite{DBLP:conf/podc/KormanV22}.

\begin{theorem}\label{cor:main}
If $185 \sqrt{n \log n} \leq k
\leq \frac{n}{2}$, then, in the synchronous parallel $k$-$\pull$, \minority\ solves the \emph{bit-dissemination} problem 
and its convergence time is $O(\log n)$ in expectation and $O(\log^2 n)$ rounds w.h.p.
\end{theorem}

We mentioned above that for all the dynamics for which we have a tight analysis there is a $\tilde \Theta(n)$ gap between the sequential and parallel convergence time, at least in the context of stateless opinion dynamics with passive-communication.

Our results show a natural example of a  dynamics for which this gap is exponential in $n$, going from $O(\log^2 n)$ to $2^{\Omega(kn)}$. Moreover, we have a natural problem, the \textit{bit-dissemination} problem, for which there is a nearly-quadratic gap between parallel passive-communication stateless solutions, for which we show an $O(\log^2 n)$ upper bound, and sequential passive-communication stateless solutions, for which an $
\Omega(n^2)$ lower bound was known.

\paragraph{Roadmap.}
The rest of this paper is organized as follows.
\cref{sec:par_overview} introduces notation and notions that will be used throughout the paper and then provides a self-contained overview of the main technical part, proving poly-logarithmic   convergence time of \minority\ for the \textit{consensus} and \textit{bit-dissemination} problems. In more detail, in \cref{subse:main_results}, we provide the high-level proofs of \cref{thm:main} and \cref{cor:main}, we state our main technical lemmas, we discuss the main ideas in our proofs, and we describe the main technical challenges that we have to overcome. 
Full proofs of all technical results above are given in \cref{sec:minority}, presented in the same order as in \cref{subse:main_results}. 
\cref{sec:lowerbound_seq} presents the proof of \cref{thm:lb_sequential}, based on the analysis of the birth-death chain of \minority.
Finally, \cref{sec::open} discusses some technical or more general questions that this work leaves open and that in our opinion deserve further investigation.

\section{\texorpdfstring{$k$-Minority}{k-Minority} in the Parallel Model }\label{sec:par_overview}

In this section, we analyze the \minority\ dynamics in the synchronous parallel communication model, over the complete graph $K_n$. We consider the case in which the sample size parameter $k$ depends on $n$, i.e. $k=k(n)$ and we prove \cref{thm:main} and \cref{cor:main} stated in the previous section.

\subsection{Notation and preliminaries}
\label{sec:notation_preliminaries}

We use the following Markov chain to describe the evolution of the \minority\ dynamics on $K_n$. Its state space is $\{0,\dots,\frac{n}{2}\} \times \{0,1\}$, and,  at every round $t\geq 0$, the state of the process at round $t$  is given  by the size $\{m_t\}_t$  and the opinion $\{\mathcal{O}_t\}_t$ of the minority.
%The state-transition matrix  of the Markov chain is the one determined by applying the local rule of \minority\ to every node of the graph, for every possible value of $m_t$. 
To get some intuition about the key quantities that govern the evolution of the process, we  introduce the following notions.

%In the following, we define 
%We say that a node $v$ is \emph{wrong} at time $t$ if $m_t \neq 0$ and $v$ adopts the opinion of the majority in round $t+1$. Therefore, if all the agents are not wrong in time $t$, in time $t+1$ all the agents (but the source agent) agree on the ex-minority opinion, reaching the consensus. In the following definition, we see that each agent $v$ can be wrong at time $t$ for two reasons.
%\atodo{it seems the above paragraph is useless?}

\begin{definition}[Wrong Nodes]\label{def:agent_wrong}
    We say node $v$ is \textit{wrong} in round $t$ if and only if: i)  $\minor_t > 0$, and ii)  $v$ adopts the majority opinion at time $t+1$. Observe that, if $m_t > 0$, the following events may cause $v$ to be wrong in round $t$:
 \begin{align*}
      \eventA_t(v) &= \text{\{$v$ samples more than $\tfrac{k}{2}$ nodes holding the minority opinion\}}\\
      \eventB_t(v) &= \text{\{$v$ samples $k$ nodes holding the majority opinion\}}.
 \end{align*}
% More precisely, we notice that if $\eventA_t(v)^C \cap \eventB_t(v)^C$ happens, then  $v$ %is not wrong at time $t$.
We define  $W_t = \sum_{v \in V}\mathds{1}_{\eventA_t(v)}$ as the r.v.~counting the  number of nodes that sample more than $\frac{k}{2}$ nodes with the minority opinion in round  $t$, and  $U_t = \sum_{v \in V}\mathds{1}_{\eventB(v)}$ as the r.v.~counting   the number of nodes that sample a unanimity of nodes holding the majority opinion in round $t$. 
\end{definition}

The following is a simple fact that relates $m_{t+1}$ to the number of wrong nodes in the previous round:
\begin{fact}\label{fa:wrong_vs_minor}
	We deterministically have that 
	\[
		m_{t+1} = 
			\begin{cases}
				W_t + U_t, &\text{if $W_t + U_t < n/2$},\\
				n - W_t - U_t, &\text{otherwise}.
			\end{cases}
	\]
    Moreover, $\mathcal{O}_{t+1} = 1 - \mathcal{O}_t$ in the first case and $\mathcal{O}_{t+1} = \mathcal{O}_t$ in the second.
\end{fact}
\begin{proof}
    In the first case, $n - W_t - U_t > \frac{n}{2}$ nodes will correctly assess and adopt the minority opinion, which will thus become majoritarian in the next round. The remaining (wrong) $W_t + U_t$ nodes will adopt the majority opinion, which will become minoritarian in round $t+1$. I.e., $m_{t+1} = W_t + U_t$ and $\mathcal{O}_{t+1} = 1 - \mathcal{O}_t$ in this case.
    The second case is similar, this time with $m_{t+1} = n - W_t - U_t \le \frac{n}{2}$ and $\mathcal{O}_{t+1} = \mathcal{O}_t$.
\end{proof}

The next lemma clarifies that the events introduced in \cref{def:agent_wrong} play distinct roles in different regimes of the process. More precisely, they are \emph{almost} mutually exclusive, in the sense that the probability that two nodes $u$ and $v$ exist, such that both $\eventA_t(u)$ and $\eventB_t(v)$ occur in any given round $t$, is essentially negligible.
\begin{lemma}
\label{lem:no_wrong}
Assume $k \geq 5 \log n$. The following holds: i) if $m_t \leq \frac{n}{2}-n\sqrt{\frac{1.5\log n}{k}}$, then $W_t = 0$ w.h.p.; ii) if $m_t \geq \frac{3n\log n}{k}$, then $U_t = 0$ w.h.p.
\end{lemma}

\cref{lem:no_wrong} is proved as part of \cref{lem:preliminaries} (points (c) and (d)), while  we next discuss some consequences. First, in every  round $t$, w.h.p. we    have either $W_t=0$ or $U_t=0$, which implies that the process $m_{t}$ can essentially be described by one of the two  random variables $U_t$ or $W_t$. 
%In more detail, if
%$m_t \leq \frac{n}{2}-n\sqrt{\frac{1.5n\log n}{k}}$,  we have $W_t = 0$ w.h.p., so that  %$m_{t+1}= \min \{U_t, n-U_t\}$. If, on the other hand, $m_t \geq \frac{3n\log n}{k}$, we %have $U_t = 0$ w.h.p., so that $m_t = \min \{W_t, n-W_t\}$. 
Moreover, it implies that if $\frac{3n\log n}{k} \leq m_t \leq \frac{n}{2}-n\sqrt{\frac{1.5n\log n}{k}}$, then \minority\ achieves consensus in round $t+1$, w.h.p. We call the above safe range  for $m_t$  the \emph{green area}. 

Thanks to the behaviour of variables $U_t$ and $W_t$ discussed above, we can identify a constant number of  areas (see \cref{fig:main_proof_structure}), each defined by a specific range for $m_t$, showing that with constant probability, \minority\ reaches the green area (and thus achieves consensus in the next round w.h.p.) within a logarithmic number of rounds: This is the   the basic strategy towards proving \cref{thm:main}. 
From this, the same bound can be shown for the {\em bit-dissemination} problem (\cref{cor:main}).

As simple as it may seem, implementing the above strategy is not straightforward. The reason is that the evolution of $m_t$ is highly non-monotonic, exhibiting a behaviour that strongly depends on the area $m_t$ belongs to. In particular, the process might jump over non-contiguous areas and might, in principle, jump back and forth between any two of them. Moreover, the presence of two fixed points for $\Expc{m_{t+1} \mid m_t}$ highlights the existence of two (unstable, as we shall see) equilibria, in which the process might get stuck. We also remark that in principle, one might attempt a proof strategy that differs from the one outlined by \cref{fig:main_proof_structure}(a). However, while this may be possible, it does not seem a trivial endeavour, since \cref{fig:main_proof_structure} actually seems to summarize ``preferred'' (or more likely) state transitions in the global Markov chain that describes the process.
%These challenges are discussed with some more detail in the following Section \ref{subse:main_results}.

\subsection{Proof of main results}\label{subse:main_results}
In this section, we prove \cref{thm:main} and \cref{cor:main}.

\paragraph{\highlight{Overview.}}
For the sake of the analysis, we partition the set of all possible configurations into a small, constant number of \textit{areas}, \highlight{that in the sequel will be named \textit{colors}}, each of them identified by a specific interval of the range $\{0,1, \ldots, \frac{n}{2}\}$ of possible values for the random variable $m_t$.
\highlight{Informally, the warmer the color, the slower the convergence   will be from the corresponding subset, as illustrated in \cref{fig:main_proof_structure}.
Formally, they are defined as follows:}
\begin{align*}
    \blue_1 &= \left\{m \in \mathbb{N}: 1 \leq m \leq   \tfrac{n\log 2}{k} -\sqrt{\tfrac{n\log 2}{k}}\right\} \\
    \red &= \left\{m \in \mathbb{N}:\tfrac{n\log 2}{k} -\sqrt{\tfrac{n\log 2}{k}} \leq m \leq \tfrac{n\log 2}{k} +\sqrt{\tfrac{n\log 2}{k}}\right\} \\
    \blue_2 &= \left\{m \in \mathbb{N}:\tfrac{n\log 2}{k} +\sqrt{\tfrac{n\log 2}{k}} \leq m \leq \tfrac{n}{k} \log \left( \tfrac{k}{4\log n}\right)\right\} \\
    \yellow &= \left\{m \in \mathbb{N}: \tfrac{n}{k} \log \left( \tfrac{k}{4\log n}\right) \leq m \leq \tfrac{3n\log n}{k}\right\} \\
    \green &= \left\{m \in \mathbb{N}: \tfrac{3n\log n}{k} \leq m \leq \tfrac{n}{2}-n \sqrt{\tfrac{2 \log n}{k}}\right\} \\
    \orange &= \left\{m \in \mathbb{N}:\tfrac{n}{2}-n \sqrt{\tfrac{2 \log n}{k}} \leq m \leq \tfrac{n}{2}\right\} 
\end{align*}

%%%%%%%%%%%%%%%%%%%%%%%%%%%%%%%%%%%%

\highlight{In order to understand the behaviour of the process in the different areas, it is useful to consider the function~$f$ represented on \cref{fig:main_proof_structure} (b).
Informally, $f$ can be seen as an approximation of $m \mapsto \Expc{m_{t+1} \mid m_t=m}$. Formally, we define it as follows:
\begin{itemize}
    \item when $m \in A = \blue_1 \cup \red \cup \blue_2$, we set
    \begin{equation*}
        f(m) = \min\{\Expc{U_t \mid m_{t}=m}, n-\Expc{U_t \mid m_t=m}\},
    \end{equation*}
    since in this range, $W_t = 0$ w.h.p.
    \item when $m \in B = \green$, we set $f(m) = 0$, since in this range, $m_{t+1}=0$ w.h.p.
    \item when $m \in C = \orange$, we set $f(m) = \Expc{W_t \mid m_t = m}$, since in this range, $U_t = 0$ w.h.p.
\end{itemize}
Using these insights, we will show that, starting from any initial configuration,  the process $m_t$ follows a short path over the colored areas, eventually landing into the green area and \highlight{then} consensus within an additional step. This typical behaviour is summarized by \cref{fig:main_proof_structure} (a).
However, this path is somewhat chaotic, in the sense that $m_t$ does not tend to evolve monotonically; instead, it can jump across non-adjacent areas over consecutive rounds.
}
%as \highlight{illustrated by}  \cref{fig:main_proof_structure} (b).
%and its caption.
%The argument summarized by  \cref{fig:main_proof_structure} ...

\begin{figure} [htbp]
    \centering
    \includegraphics[width=\textwidth]{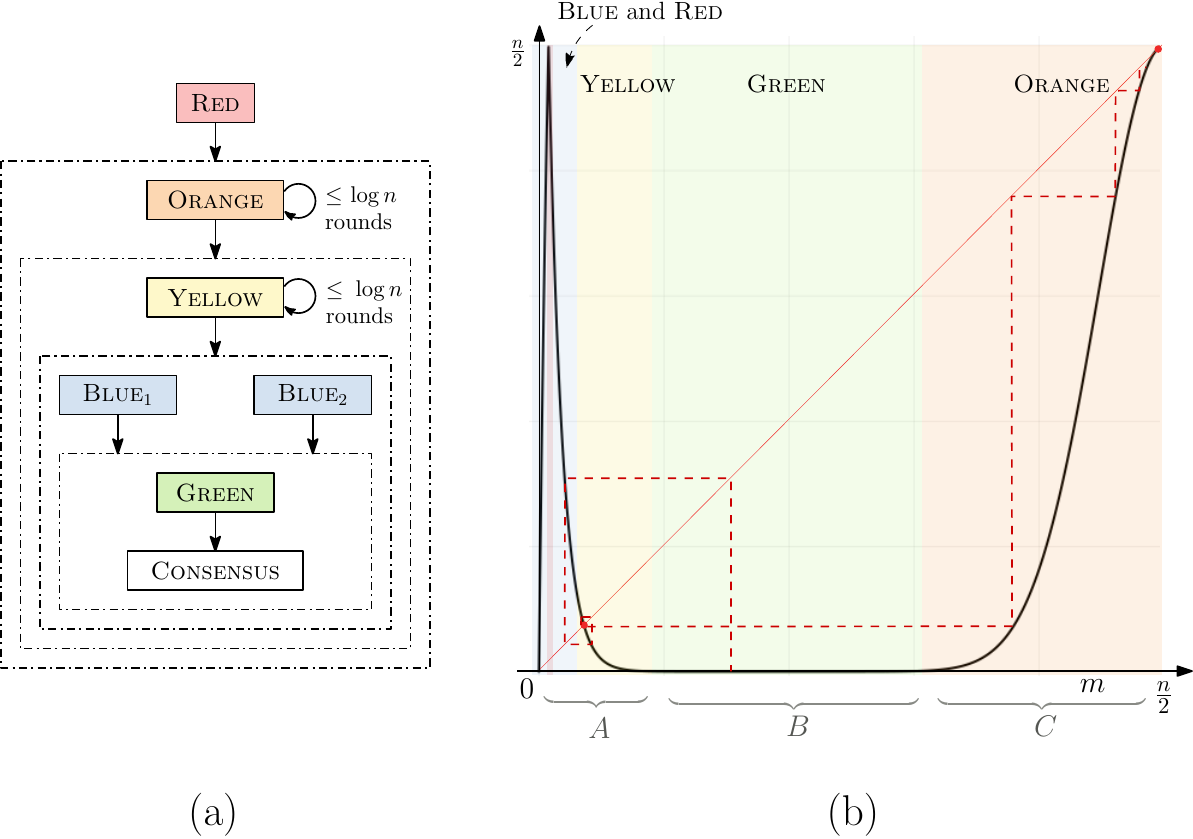}
    %%% OLD CAPTION (before review)
    %\caption{\small (a): Roadmap for the  proof of \cref{thm:main}. Unless explicitly indicated, transitions occur with constant probability in a constant number of rounds. Next to self-loops, we give the number of rounds the process spends in corresponding areas. (b): Plot of  function $f(m)\approx \Expc{m_{t+1} \mid m_t = m}$. The x-axis is partitioned according to the colored areas and the three intervals $A = \bigg[0,\frac{3n\log n}{k}\bigg]$, $B=\bigg[\frac{3n\log n}{k},\frac{n}{2}-n\sqrt{\frac{2n\log n}{k}}\bigg]$ and $C= \bigg[\frac{n}{2}-n\sqrt{\frac{2n\log n}{k}}, \frac{n}{2}\bigg]$. We define  $f(m)$ so that $f(m)= \min\{\Expc{U_t \mid m_{t}=m}, n-\Expc{U_t \mid m_t=m}\}$ if $m \in A$, since $W_t = 0$ w.h.p. and thus $m_{t+1} = \min \{U_t,n-U_t\}$ (w.h.p.) in this range. We define $f(m)=0$ whenever $m \in B$, since $m_{t+1}=0$ w.h.p., if $m_t \in B$. Finally, we have $U_t$ w.h.p. and thus $m_{t+1} = \min\{W_t,n-W_t\}$, whenever $m_t\in C$. Accordingly, we define $f(m)=\Expc{W_t \mid m_t = m}$ if $m\in C$. The red circles are the fixed points of     $f(m)$, i.e. $m_1 = f(m_1)$ and $m_2 = f(m_2)=\frac{n}{2}$, while the red dashed track represent some iteration of $f(m)$, starting with a configuration close to $\frac{n}{2}$. While some relevant details of our analysis are missed, the behavior of $f(m)$ highlights the non-monotonic, chaotic behavior of  the underlying process.}
    %%% SHORTER CAPTION
    \caption{\highlight{\small {\bf (a)} Roadmap for the  proof of \cref{thm:main}. Unless explicitly indicated, transitions occur with constant probability in a constant number of rounds. Next to self-loops, we give the number of rounds the process spends in corresponding areas. {\bf (b)} Function~$f$, satisfying $f(m) \approx \Expc{m_{t+1} \mid m_t=m}$, and the colored areas. The red circles are the fixed points of $f(m)$, i.e. $m_1 = f(m_1)$ and $m_2 = f(m_2)=\frac{n}{2}$. The red dashed track represent some iteration of $f(m)$, starting with a configuration close to $\frac{n}{2}$ -- illustrating the non-monotonic, chaotic behavior of the process.}}
\label{fig:main_proof_structure}
\end{figure}

\paragraph{\highlight{Intermediate Results.}}
\highlight{In order to formalize the aforementioned arguments, we state a few intermediate lemmas, which we will prove in subsequent sections.
}

\begin{lemma}[Red Area] \label{lem:out_of_red}
   For any $t \geq 1$ and $m \in [\frac{n}{2}]$, $\Prob{m_{t+1} \not \in \red \mid m_t =m} \geq c$, for a suitable positive constant $c$.
\end{lemma}
The proof of \cref{lem:out_of_red} uses a simple variance argument relying on the modest width of the red area, \highlight{and is deferred to \Cref{sec:preliminary_lemmas}}. The next lemma, \highlight{however,} poses some technical challenges.

\begin{lemma}[Orange Area] \label{lem:out_of_orange} For any $t_0 \geq 1$ and $m_{t_0} \in \orange$, let \[T = \inf \{ s>t_0: m_{s} \not \in \orange  \}.\] Then, for a suitable positive constant $c$,
\[\Prob{\{ \highlight{T- t_0} \leq \log n\} \cap \{ \highlight{m_T} \not \in \red \cup \orange\} \mid m_{t_0} = m} \geq c.\]
\end{lemma}
The proof of \cref{lem:out_of_orange} is given in \cref{subse:orange}.
The first problem is that the \orange\ area contains the fixed point $m_t = \frac{n}{2}$. To prove that the process leaves this area within $O(\log n)$ rounds, we need to show that it tends to drift away from $\frac{n}{2}$. In \cref{claim:proba_minority_wins} and \cref{claim:out_of_orange}, we indeed prove that $m_{t+1}\le \frac{n}{2} - 1.5\frac{\alpha}{\sqrt{n}}$ with probability at least $1 - e^{-0.1\alpha^2} - \frac{1}{n^2}$, whenever $m_t = \frac{n}{2} - \frac{\alpha}{\sqrt{n}}$. This is an exponentially increasing drift occurring with increasing probability. This fact allows us to show that the process leaves the \orange\ area with constant probability, within $O(\log n)$ rounds. 

The second problem we face is that, upon leaving the \orange\ area, the process might in principle fall back to the \red\ one. In order to prove that this event does not occur with constant probability, we need to proceed with care, in order to somehow capture two different conditions under which $m_t$ may have left the \orange\ area, each of which needs a distinct characterization.

\begin{lemma}[Yellow Area] \label{lem:out_of_yellow}
For any $t_0 \geq 1$ and $m_{t_0} \in \yellow$, let \[T = \inf \{s >t_0: m_{s} \not \in \yellow\}.\] Then, for a suitable positive constant $c$,
\[\Prob{\{T \leq \log n \} \cap \{m_{t_0 + T} \in \blue_1 \cup \blue_2 \cup \green\} \mid m_{t_0}=m} \geq c.\]
\end{lemma}
The proof of \cref{lem:out_of_yellow}, given in \cref{subse:orange}, proceeds along lines similar to those of \cref{lem:out_of_orange}, albeit with a few differences that do not allow a simple reuse of the arguments given for the \orange\ case. Again, we have a fixed point $\overline{m}$ in the \yellow\ area, but one that we do not compute exactly. This time, the roles of $W_t$ and $U_t$ are reversed, since we are in a regime in which $W_t = 0$ w.h.p. Moreover, while in the previous case the drift was towards decreasing values for $m_t$, the values of $m_t$ are considerably smaller (i.e., $O(\sqrt{n})$) and the drift potentially harder to characterize in the \yellow\ area. In particular, while drifting away from $\overline{m}$, $m_t - \overline{m}$ can change sign over consecutive rounds. For this reason, it proved necessary to quantify drift using $|m_t - \overline{m}|$ in the proof of \cref{lem:expectation_in_yellow2}. Finally, some extra care is needed to address the fact that this time, we are interested in the probability that, upon leaving the \yellow\ area, the process avoids both the \red\ and \orange\ ones.

The two remaining cases address the \blue~and the \green~area. Their analysis is easy and relies on simple concentration arguments.

\begin{lemma}[Blue Area] \label{lem:out_of_blue}
For any $t \geq 1$ and $m \in \blue_1 \cup \blue_2$, $\Prob{m_{t+1} \in \green \mid m_t = m}\geq c$, for a suitable positive constant $c$.
\end{lemma}

\begin{lemma}[Green Area]\label{lem:out_of_green}
For any $t \geq 1$, $m \in \green$ and $\ell \in \{0,1\}$, 
\[\Prob{\{m_{t+1}=0\} \cap \{\mathcal{O}_{t+1} = 1-\ell\} \mid m_t = m, \mathcal{O}_t = \ell}\geq c,\]
for a suitable positive constant $c$.\footnote{Note that if $m_t = 0$, $\mathcal{O}_t$ is the opinion opposite to the majority.}
\end{lemma}

\begin{proof} [Proof of \cref{thm:main}]
    %See \cref{fig:main_proof_structure} for an overview.
    Recall that by definition of the protocol, for every agent $i$, if all samples collected by $i$ in round~$t$ are equal to its opinion, then agent~$i$ maintains the same opinion in round~$t+1$. Therefore, if~$m_t = 0$, we have $m_s = 0$ for every~$s>t$, i.e., consensus has been reached permanently.
    Moreover, Lemmas~\ref{lem:out_of_red} to~\ref{lem:out_of_green} prove the existence of two constants $c_1,c_2 > 0$, such that for every $t \geq 0$ and $m \in [\frac{n}{2}]$,
    \begin{equation} \label{eq:convergence_proba}
        \Prob{m_{t+c_1 \log n} = 0 \mid m_t =m} > c_2.
    \end{equation}
    Let~$T = \inf \{ s \in \mathbb{N}, m_s = 0 \}$ and 
    %Eq.~\eqref{eq:convergence_proba} can be rewritten as $\Prob{T<c_1 \log n} > c_2$.
    $x \in \mathbb{N}$. Since \eqref{eq:convergence_proba} holds for every~$t$ and since $\{m_t\}_t$ is a Markov chain, we have $\Prob{T < x \cdot c_1 \log n} > 1 - (1- c_2)^x$.
    We can thus take a constant $c_3$ such that, for $x \geq c_3 \log n$, we have $
        \Prob{T < c_1 c_3  \log^2 n} > 1-\frac{1}{n^2}$,
    which concludes the proof of \cref{thm:main}.
\end{proof}

The following Lemma, whose proof is given in \cref{subse:bit_dissemination}, yields \cref{cor:main}. 
\begin{lemma}[Bit-Dissemination Problem]\label{lem:wrong_consensus}
    If $\ell$ is the opinion of the source agent, then \[\Prob{m_{t+2} = 0 \mid m_{t}=1, \mathcal{O}_t = \ell}\geq c,\] for a suitable positive constant $c$.
\end{lemma}

\begin{proof}[Proof of \cref{cor:main}]
Consider the \minority \ dynamics in the presence of a source agent. In this case, 
\cref{thm:main} implies the existence two constants $c_1,c_2>0$ such that, for every $t$ and $m \in [\frac{n}{2}]$, $\Prob{m_{t+c_1\log n}\leq 1 \mid m_t = m} >c_2$. Moreover,   \cref{lem:wrong_consensus} implies a constant $c_3>0$ for which
\[\Prob{m_{t+2+c_1\log n}=0 \mid m_t = m}>c_3.\]
Setting $T = \inf\{s \in \mathbb{N}, m_s = 0 \}$ and arguing as in the proof of \cref{thm:main}, we have $m_T = 0$ w.h.p.~for some~$T=O(\log^2 n)$.
\end{proof}

\section{\texorpdfstring{$k$-Minority}{k-Minority} in the Parallel Case: Technical Details}\label{sec:minority}
We gave an overview of the proofs of \cref{thm:main} and \cref{cor:main} in Section \ref{sec:par_overview}. This section contains the proofs of all technical lemmas that were used there, some of which highlight important properties of the \minority dynamics, some of which are not trivial to capture, due to the chaotic nature of the process.
\subsection{Preliminary lemmas}
\label{sec:preliminary_lemmas}
We begin by proving a set of simple results describing the behavior of the random variables $W_t$ and $U_t$, which will be used throughout the analysis. 
Our proofs often use two standard probabilistic  tools,  namely, \textit{Chernoff Bound} and the \textit{Reverse Chernoff Bound},  in the versions presented in  \cref{thm:additive_chernoff} and \cref{lem:add_reverse_chernoff} of the Appendix, respectively. 
The following is a general statement that also includes the results singled out in \cref{lem:no_wrong} of \cref{sec:notation_preliminaries}. \highlight{Its proof is deferred to \Cref{app:preliminaries_proof}.}

\begin{lemma}
    \label{lem:preliminaries}
    Assume $k \geq 5 \log n$. Then, for any $m \in \{0,\dots, \frac{n}{2}\}$, we have the following results
    \begin{enumerate}[(a)]
        \item $ \frac{n}{4}e^{-4k\left(\frac{1}{2}-\frac{m}{n}\right)^2} \leq \Expc{W_t \mid m_t=m} \leq n e^{-2k\left(\frac{1}{2}-\frac{m}{n}\right)^2}$;
        \item $\Expc{U_t \mid m_t = m} = n\pa{1-\frac{m}{n}}^k$; 
        \item For any fixed $m_t = m$ with $m \geq \frac{3n\log n}{k}$, $U_t = 0$ w.h.p.
        \item For any fixed $m_t = m$ with  $m \leq \frac{n}{2}-n\sqrt{\frac{1.5 \log n}{k}}$, $W_t = 0$ w.h.p.
        \item For any fixed $m_t = m$ with $m \leq \frac{n}{3k}$, $U_{t} \geq 0.6n$ and $\minop_{t+1} = \minop_t$ w.h.p.
        \item For any fixed $m_t = m$ with  $\frac{2n}{k} \leq m \leq \frac{n}{2}-n\sqrt{\frac{1.5 \log n}{k}}$, $U_t\leq 0.2n$ and $\minop_{t+1} = 1-\minop_t$ w.h.p.
    \end{enumerate}
\end{lemma}

The following lemma is a simple concentration result applied to the variables $U_t$, $n-U_t$ and $W_t$, and gives necessary conditions so that the value of these random variables is close to their expectation.

\begin{lemma}
Consider an interval $I=[a,b]\cap \mathbb{N}$ such that $b \geq 4a$ and $b \geq 32\log n$ and fix a value $m_t = m$.\footnote{To the purpose of our analysis, we are interested in the case $I \subseteq [0,n]$, but the result is more general.} Let $X$ be any of the random variables in $\{U_t, n-U_t, W_t\}$. For any $m \in \{0, \dots, \frac{n}{2}\}$, if $\Expc{X \mid m_t=m} \in I$, then  $X\in I$ with constant probability. Moreover, considered any constant $\gamma\geq \frac{3}{2}$, if $a \geq 36\log n$ then $X\in \gamma I$ w.h.p., where $\gamma I=[\frac{a}{\gamma},\gamma b] \cap \mathbb{N}$.
\label{lem:mt_given_expectation}
\end{lemma}
\begin{proof}
We first remark that $X$ can always be written as the sum of $n$ independent Bernoulli random variables. Let $M_I$ be the median point of $I$, i.e. $M_I = a+\frac{b-a}{2}$. Note that, since $b \geq 4a$ and $b \geq 32 \log n$, we have $M_I \geq \frac{b}{2}\geq 16\log n$.
We consider two cases and, without loss of generality, we prove the lemma for $\gamma=\frac32$. 

First, assume that $\Expc{X \mid m_t=m}\leq M_I$. In this case, the use of Chernoff's Bound implies that
$\Prob{X\leq 1.5M_I \mid m_t = m } \geq 1-e^{-M_I/8} \geq 1- \frac{1}{n^{2}}$,
since $M_I \geq 16\log n$ and $b \geq 4a$, $1.5M_I \leq b$. 
If  $a \geq 36 \log n$, we further have $\Prob{X \geq \frac23\Expc{X \mid m_t = m} \mid m_t = m} \geq1- e^{-a/18} \geq 1-\frac{1}{n^{2}}$. If no condition on $a$ is given, the reverse Chernoff bound guarantees the existence of a constant $c>0$, such that $\Prob{X \geq a} \geq \Prob{X \geq \Expc{X \mid m_t = m} \mid m_t = m} \geq c$, where the first inequality follows from the assumption $\Expc{X \mid m_t = m} \in I$.
Therefore, in the absence of constraints on $a$, we have 
$\Prob{X \in I \mid m_t = m} \geq \Prob{a\leq X \leq 1.5 M_I \mid m_t = m}>c-\frac{1}{n^2}$. If, on the other hand, $a \geq 36 \log n$, then $\Prob{X \in \frac32 I \mid m_t = m} \geq \Prob{\frac23\Expc{X \mid m_t = m \mid m_t = m}\leq X\leq 1.5M_I \mid m_t = m} \geq 1-\frac{2}{n^2}$.

Next, assume that $\Expc{X \mid m_t = m}\geq M_I$. Similarly to the previous case,  Chernoff bound gives $\Prob{X \geq 0.5M_I} \geq 1-e^{-M_I / 8 } \geq 1-\frac{1}{n^2}$ and, since $b \geq 4a$, we have $0.5M_I \geq a$. Moreover, since $b \geq 32 \log n$, we also have  $\Prob{X \leq \frac32b} \geq 1-e^{-M_I/8} \geq 1-\frac{1}{n^2}$.
In addition, the reverse Chernoff bound gives the existence of a constant $c'>0$, such that $\Prob{X \leq b} \leq\Prob{X \leq \Expc{X \mid m_t=m}}>c'$, where the first inequality follows from the fact that $\Expc{X \mid m_t = m} \in I$. In conclusion, we have  $\Prob{X \in I } \geq \Prob{0.5M_I\leq X\leq \Expc{X \mid m_t=m}} \geq c'-\frac{1}{n^2}$ and $\Prob{X \in \frac32 I \mid m_t = m} \geq \Prob{0.5M_I \leq X \leq \frac32 b \mid m_t = m} \geq 1-\frac{2}{n^2}$. 
\end{proof}

\subsection{The red area}

The width of the Red area is small enough that given $m_t$, $m_{t+1}\not \in \red$ with at least constant probability.
The following result formalizes this idea and relies on a simple anti-concentration result.

\begin{lemma} \label{lem:red_to_blue}
Let $X$ be any random variable in $\{U_t, n-U_t, W_t, n-W_t\}$.  For any value $m \in \{0, \dots, \frac{n}{2}\}$, if $\Expc{X \mid m_t = m} \in \red$ then $X \in \blue_1 \cup \blue_2$ with constant probability.
\end{lemma}
\begin{proof}
We first remark that $X$ can always be written as a sum of $n$ i.i.d. Bernoulli random variables and, since $\Expc{X\mid m_t=m} \in \red$, we have $\Expc{X \mid m_t = m} \leq \frac{n}{2}$, thus meeting the assumptions of the reverse Chernoff bound applied to $X$  under the conditioning $m_t = m$. Moreover, since $k \leq 0.5n$, we have $\frac{n \log 2}{k}\geq 1$. By definition, the \red \ area has width $2\sqrt{\frac{n \log 2}{k}}$ approximately, so that $\Expc{X \mid m_t=m}$ lies within distance at most $\sqrt{\Expc{X \mid m_t=m}}$ from the border of $\red$. For this reason, the reverse Chernoff bound implies that $X$ can deviate at least $\sqrt{\Expc{X \mid m_t = m}}$ from its expectation, with constant probability. This concludes the proof. 
\end{proof}

The proof of \cref{lem:out_of_red} is a simple application of previous results.

\begin{proof}[Proof of \cref{lem:out_of_red}]
Assume that, for $m_t = m$, $\Expc{X \mid m_t=m} \in \red$ for a random variable $X \in \{U_t, n-U_t,W_t, n-W_t\}$. Then, \cref{lem:red_to_blue} implies $X \in \blue_1 \cup \blue_2$ with constant probability. This, together with \cref{lem:preliminaries}(c) and \cref{lem:preliminaries}(d), implies
$m_{t+1} \in \blue_1 \cup \blue_2$ with constant probability.

Next, suppose that $\Expc{X \mid m_t =m} \not \in \red$ for every choice of $X \in \{U_t,n-U_t,W_t,n-W_t\}$. Consider $\Bar{X}=\mathrm{argmin}_{X\in \{U_t,n-U_t,W_t,n-W_t\}} \Expc{X \mid m_t = m}$: we obviously have $\Expc{\Bar{X} \mid m_t = m} \leq \frac{n}{2}$ and $\Expc{\Bar{X} \mid m_t = m} \not \in \red$. We can thus apply \cref{lem:mt_given_expectation} to the intervals $I_1 = \blue_1$ and $I_2 = \blue_2 \cup \yellow \cup \green \cup \orange$. This, together with \cref{lem:preliminaries}(c) and \cref{lem:preliminaries}(d), proves that $\Prob{m_{t+1} \not \in \red\mid m_t = m} \geq \Prob{\Bar{X} \in I_1 \cup I_2 \mid m_t = m} > c$, for a constant $c>0$.
\end{proof}

\subsection{The orange area}\label{subse:orange}

The goal of this section is to prove \cref{lem:out_of_orange}. In this area, we have the presence of the fixed point $\frac{n}{2}$ for the expectation of the process, so that $\Expc{m_{t+1} \mid m_t = \frac{n}{2}}=\frac{n}{2}$. In the remainder of this section, we prove that the aforementioned fixed point is unstable. The main ingredients of our proof are the following: i) first, even if $m_t = \frac{n}{2}$, $m_{t+1}$ can deviate by  $\Omega(\sqrt{n})$ from its conditional expectation ($\frac{n}{2}$) thanks to the variance of the process; ii) assuming instead that $m_t = \frac{n}{2}-\Omega(\sqrt{n})$, we are able to show a drift in the process, proving that with constant probability, it leaves the \orange~area within an at most logarithmic number of steps.
The following lemma quantifies the drift mentioned in point ii) above.

\begin{lemma} \label{claim:proba_minority_wins}
    Let $m = \frac{n}{2}-\alpha \sqrt{n}$, where $1 \leq \alpha \leq \frac{\sqrt{n}}{4}$. Then, $\Expc{W_{t} \mid m_t = m} \leq \frac{n}{2}-1.8\alpha \sqrt{n}$.
\end{lemma}

\begin{proof}
Fix an agent $i$ and consider  the random variable $X$ indicating the number of agents sampled by $i$ and holding the minority opinion. Our goal is to prove that $\Prob{X >\frac{k}{2} \mid m_t = \frac{n}{2}-\alpha \sqrt{n}}\leq \frac{1}{2}-\frac{\alpha}{1.1\sqrt{n}}$. For simplicity, we omit the conditioning in the rest of the proof. We have from \cref{claim:centralbinomial},
    \begin{align}
    \Prob{X =  \left\lceil \tfrac{k}{2}\right\rceil } \leq \Prob{X =  \left\lfloor \tfrac{k}{2}\right\rfloor } 
 \leq \binom{k}{\lfloor \tfrac{k}{2}\rfloor}\left(\frac{1}{2}\right)^k \leq \frac{4^{k/2}}{\sqrt{\pi(k/2)}\cdot 2^k} \leq \frac{1}{\sqrt{k}}.
    \label{eq:X_exactsymmetry}
    \end{align}
    Denote by $Y$ the number of agents sampled by $i$ that hold the majority opinion, so that $Y=k-X$. 
    Since $X\sim\text{Bin}\left(k,1-\frac{m}{n}\right)$ and $Y\sim\text{Bin}\left(k,\frac{m}{n}\right)$, it is easy to see that
    \begin{align*}
        \Prob{X \geq \left\lceil\tfrac{k}{2}\right\rceil + 1 } \leq \left(\frac{\tfrac{m}{n}}{1-\tfrac{m}{n}}\right)^2\Prob{Y \geq \left\lceil\tfrac{k}{2}\right\rceil + 1}.
    \end{align*}
    Letting $\gamma = \frac{2\alpha}{\sqrt{n}}$ and $\beta = \left(\frac{1-\gamma}{1+\gamma}\right)^2$, we can rewrite the inequality above as follows
    \begin{align*}
        \Prob{X \geq \left\lceil\tfrac{k}{2}\right\rceil + 1 } \leq  \beta \cdot \Prob{Y \geq \left\lceil\tfrac{k}{2}\right\rceil + 1}.
    \end{align*}
Since $X = k-Y$, $\Prob{Y \geq x}= \Prob{X\leq k-x}$ for every $x \in [k]$, so that
\begin{align*}
    \Prob{X \geq \left\lceil\tfrac{k}{2}\right\rceil + 1}  \leq \beta \cdot \left(1-\Prob{X \geq \left\lceil\tfrac{k}{2}\right\rceil} - \Prob{X 
 = \left\lfloor\tfrac{k}{2}\right\rfloor} \right).
\end{align*}
Moreover, \cref{eq:X_exactsymmetry} and the above inequality imply
\begin{align*}
  \Prob{X \geq \left\lceil\tfrac{k}{2}\right\rceil } &\leq \beta \cdot \left(1-\Prob{X \geq \left\lceil\tfrac{k}{2}\right\rceil } \right) - \beta \cdot \Prob{X 
 = \left\lfloor\tfrac{k}{2}\right\rfloor}+\Prob{X 
 = \left\lceil\tfrac{k}{2}\right\rceil}
 \\ &\leq \left(\frac{1-\gamma}{1+\gamma}\right)^2  \cdot \left(1-\Prob{X \geq \left\lceil\tfrac{k}{2}\right\rceil } \right) + \Prob{X 
 = \left\lfloor\tfrac{k}{2}\right\rfloor}\left(1-\beta\right)
 \\ & \leq \beta  \cdot \left(1-\Prob{X \geq \left\lceil\tfrac{k}{2}\right\rceil } \right) + \frac{1}{\sqrt{k}}\left(1-\beta\right)
 \\ & \leq \beta  \cdot \left(1-\Prob{X \geq \left\lceil\tfrac{k}{2}\right\rceil } \right) + \frac{6\gamma}{\sqrt{k}},
\end{align*}
where the last inequality holds whenever $\gamma < 0.3$ (which, in our case, is true since $\gamma=\frac{2\alpha}{\sqrt{n}}$ and $\alpha <\frac{\sqrt{n}}{4}$).
Finally, simple calculus yields
\begin{align*}
\Prob{X \geq  \left\lceil\tfrac{k}{2}\right\rceil} \leq \frac{\left(\frac{1-\gamma}{1+\gamma}\right)^2+\frac{6\gamma}{\sqrt{k}}}{1+\left(\frac{1-\gamma}{1+\gamma}\right)^2} \leq \frac{1}{2}-\frac{\gamma}{1.1} + \frac{6\gamma}{\sqrt{k}} \leq \frac{1}{2}-0.92\gamma \leq \frac{1}{2}-1.8\cdot \frac{\alpha}{\sqrt{n}}.
\end{align*}
where the first inequality holds whenever $\gamma<0.3$. Hence, noting that $\Expc{W_{t} \mid m_t = m} = n \cdot \Prob{X \geq  \left\lceil\tfrac{k}{2}\right\rceil}$ the claim follows.
\end{proof}

\begin{lemma}
For any $t_0 \geq 1$ and $m \in \orange$, let
\[T = \inf \{s>t_0: m_s \not \in \orange\},\].
Then, $\Prob{\highlight{T - t_0 \leq \log n} \mid m_{t_0}=m } \geq c$, for a suitable positive constant $c$.
\label{claim:out_of_orange}
\end{lemma}

\begin{proof}
    For every~$t$, define $\alpha_t$ so that $m_t = \frac{n}{2}-\alpha_t \sqrt{n}$, i.e., $\alpha_t =\left(\frac{n}{2}-m_t\right)\frac{1}{\sqrt{n}}$. Note that, if $m_t \in \orange$, then $\alpha_t \in \left[0,\sqrt{\frac{2n\log n}{k}}\right]$. We can prove the following:
    \begin{enumerate}[(i)]
        \item If $\alpha_{t_0} \leq 3$, then $\alpha_{t_0+1} \geq 3$ with constant probability.
        \item For every $t>t_0$ and $3 \leq \alpha \leq \frac{\sqrt{n}}{4}$, we have $\Prob{\alpha_{t+1}>1.5\alpha \mid \alpha_t=\alpha} \geq 1-e^{-0.1\alpha^2}-\frac{1}{n^2}$.
    \end{enumerate}
We prove (i) and (ii) below, after showing that they imply \cref{claim:out_of_orange}.

Assume (i) and (ii) hold. We define the events $\mathcal{E}_t = \{t \geq T\} \cup \{\alpha_{t+1}>1.5\alpha_t\}$, $\mathcal{F}= \{\alpha_{t_0+1}\geq 3\}$ and $\mathcal{E} = \bigcap_{t=t_0+1}^{t_0+1+\log n}\mathcal{E}_t \bigcap \mathcal{F}$. We note that, conditioning on $m_{t_0}=m$, we have $\{T \leq \log n\} \subseteq \mathcal{E}$. Indeed, this follows since, if $\alpha_{t_0+1}\ge 3$ and $\alpha_{t+1} > 1.5\alpha_t$ for every $t = t_0 + 1,\ldots, t_1$, then a straightforward calculation shows that $\alpha_{t_1}\ge\sqrt{\frac{2n\log n}{k}}$ for $t_1 < \log n$. Recalling that (i) implies $\Prob{\mathcal{F}}>c_1$ for a suitable constant $c_1$ we have
\begin{equation}
    \label{eq:out_orange}
    \Prob{\highlight{T - t_0 \leq \log n} \mid m_{t_0} = m}\ge\Prob{\mathcal{E} \mid m_{t_0} = m}\ge c_1\prod_{t=t_0+1}^{t_0+1+\log n}\Prob{\mathcal{E}_t \mid \cap_{j=t_0+1}^{t-1}\mathcal{E}_j \cap \mathcal{F}, m_{t_0} = m}.
\end{equation}
Now note that, from our definition of the $\mathcal{E}_t$'s, we have $\Prob{\mathcal{E}_t \mid \cap_{j=t_0+1}^{t-1}\mathcal{E}_j\cap \mathcal{F}, m_{t_0} = m} = 1$ if $T = t'$ for some $t'\in\{t_0+1,\ldots , t-1\}$, while $\Prob{\mathcal{E}_t \mid \cap_{j=t_0+1}^{t-1}\mathcal{E}_j\cap \mathcal{F}, m_{t_0} = m}\ge \left(1-e^{-0.1\alpha^2}-\frac{1}{n^2}\right)$ if $t' < T$ for $t'\in\{t_0+1,\ldots , t-1\}$ and $\alpha_t = \alpha$.
To finish the proof, define $\beta_t = 3\cdot (1.5)^{t-t_0-1}$. Then, for every $t$
\begin{align*}
    \Prob{\mathcal{E}_t \mid \cap_{j=t_0+1}^{t-1}\mathcal{E}_j\cap \mathcal{F}, m_{t_0} = m} &\geq \Prob{\alpha_{t+1} > \beta_{t+1} \mid \cap_{j=t_0+1}^{t-1}\alpha_{j+1} > 1.5\alpha_j, t<T} \\
    &\geq 1-e^{-0.1\beta_t^2}-\frac{1}{n^2},
\end{align*}
where the first inequality follows since $\cap_{j=t_0+1}^{t-1}\alpha_{j+1} > 1.5\alpha_j$ implies $\alpha_t > \beta_t$, while the second follows from this fact and from (ii). Finally, \cref{eq:out_orange} and the inequality above give 
\begin{align*}
    &\Prob{\highlight{T - t_0 \leq \log n} \mid m_{t_0} = m}\ge\Prob{\mathcal{E} \mid m_{t_0} = m} \geq c_1\prod_{t=t_0+1}^{t_0+1+\log n}\left(1-e^{-0.1\beta_t^2}-\frac{1}{n^2}\right) \\ &\geq c_1\left(1-\sum_{t=t_0+1}^{t_0+1+\log n}e^{-0.1\beta_t^2} -\frac{1}{n}\right)\geq c_1\left(1-e^{-0.1\cdot 9} -\sum_{t=1}^{+\infty}3^{-t}-\frac{1}{n}\right)\geq c,
\end{align*}
for a suitable positive constant $c$, with the third inequality following from the definition of $\beta_t$.

\paragraph{Proof of (i)}
We consider two cases. Assume first that $\Expc{W_t \mid \alpha_{t_0}=\alpha} \leq \frac{n}{2}-3\sqrt{n}$. In this case, \cref{lem:mt_given_expectation} implies that with constant probability, we $W_{t_0} \leq \frac{n}{2}-3\sqrt{n}$, whence $m_{t_0+1} \leq \frac{n}{2}-3\sqrt{n}$ and  $\alpha_{t_0+1} \geq 3$. Assume next that $\frac{n}{2}-3\sqrt{n}\leq \Expc{W_t \mid \alpha_{t_0}=\alpha} \leq \frac{n}{2}$. In this case, we apply the reverse Chernoff Bound to show that with constant probability, $W_{t_0+1}$ deviates at least $3\sqrt{\Expc{W_t \mid \alpha_t=\alpha}}$ from its expectation, thus giving $\alpha_{t_0+1}\geq 3$.

\paragraph{Proof of (ii)} By definition, we have that $m_t = \frac{n}{2}-\alpha_t \sqrt{n}$. Now, \cref{claim:proba_minority_wins} implies
\begin{equation*}
    \Expc{W_t \mid \alpha_t = \alpha }\leq \frac{n}{2}-1.8\alpha\sqrt{n}.
\end{equation*}
Therefore, Chernoff's bound (\cref{thm:additive_chernoff}), \cref{lem:preliminaries}(c) and a union bound give
\begin{equation*}
    \Prob{U_t = 0 \text{ and }W_t\leq \frac{n}{2} - 1.5 \, \alpha_t \sqrt{n} \mid \alpha_t = \alpha} \geq 1-e^{-0.1\alpha_t^2}-\frac{1}{n^2}.
\end{equation*}
Hence, with the above probability we have $m_{t+1}=W_t$, so that
\[\Prob{\alpha_{t+1}>1.5\alpha \mid \alpha_t = \alpha} = \Prob{m_{t+1}\leq \frac{n}{2}-1.5\alpha \sqrt{n} \mid \alpha_t = \alpha} \geq 1-e^{-0.1\alpha_t^2}-\frac{1}{n^2}.\]
\end{proof}

We now proceed with the proof of \cref{lem:out_of_orange}.
\iffalse
\begin{proof} [Proof of \cref{lem:out_of_orange}]
For $m\in\orange$, the event $\{T \leq \log n\} \cap \{m_{t_0 + T} \not \in \red \cup \orange \mid m_{t_0} = m\}$
is not realized if the following event occurs:
\[
    \{T > \log n\}\cup\mathcal{E\mid m_{t_0} = m},
\]
with $\mathcal{E} = \{\exists i = 1,\dots, \log n: (\cap_{j=0}^{i-1}m_{t_0+j}\in\orange)\cap(m_{t_0+i}\in\red)\}$. We have:
\begin{align*}
    &\Prob{\mathcal{E}\mid m_{t_0} = m}\le\sum_{i = 1}^{\log n}\Prob{(\cap_{j=0}^{i-1}m_{t_0+j}\in\orange)\cap(m_{t_0+i}\in\red) \mid m_{t_0} = m}\\
    &\le \sum_{i = 1}^{\log n}\Prob{(m_{t_0+i-1}\in\orange)\cap(m_{t_0+i}\in\red) \mid m_{t_0} = m}\\
    &\le \sum_{i = 1}^{\log n}\Prob{m_{t_0+i}\in\red \mid m_{t_0+i-1}\in\orange, m_{t_0} = m}
\end{align*}
\end{proof}
\fi

\begin{proof} [Proof of \cref{lem:out_of_orange}]
Consider the stopping time $T'$ defined as follows:
   \[T'=\inf \{s>t_0: \Expc{W_{s} \mid m_s} \not \in \orange  \text{ or } m_{s} \not \in \orange\}.\]
   Roughly speaking, $T'$ is the first time the expectation of the process, or the process, leaves the \orange~area. Of course, $T'\leq T$. Considering $T'$ instead of $T$ is useful to prove that, at time $T=T'+1$, $m_T\not\in\orange\cup\red$ with constant probability.
Consider the events
\[A_t = \{m_t \in \orange\}\cap \{ \Expc{W_t \mid m_t} \in \orange\} \quad \text{ and } \quad B_t = \{\tfrac{n}{3} \leq W_t \leq \tfrac{3n}{4}\} \cap \{U_t = 0\},\]
and let
\[\mathcal{E}_t = (A_t)^C \cup B_t \cup \{t \geq T'\}, \qquad \mathcal{E}_1 = \bigcap_{t=t_0}^{t_0+\log n}\mathcal{E}_t \quad \text{ and } \quad \mathcal{E}_2 = \{\highlight{T' - t_0 \leq \log n}\}.\]
We have:
\begin{align*}
    \Prob{\mathcal{E}_1^C\mid m_{t_0}=m} &\leq \sum_{t=t_0}^{t_0+\log n}\Prob{\mathcal{E}_t^C \mid m_{t_0}=m} \\&\leq \sum_{t=t_0}^{t_0+\log n}\Prob{(B_t)^C \mid A_t, t < T', m_{t_0}= m} \\&\leq \frac{2\log n}{n^2},
\end{align*}
where the last inequality follows from \cref{lem:mt_given_expectation} applied to the interval $I=\orange$, by observing that $\frac{3}{2}I \subseteq [\frac{n}{3}, \frac{3n}{4}]$ and from \cref{lem:preliminaries}(c).
Since $T'\leq T$,
\cref{claim:out_of_orange} shows the existence of a constant $c_2>0$ such that $\Prob{\mathcal{E}_2 \mid m_{t_0}=m}>c_2$, whence $\Prob{\mathcal{E}_1\cap \mathcal{E}_2 \mid m_{t_0}=m}>c_2-\frac{1}{n}>0$. We assume that the event $\mathcal{E}_1\cap \mathcal{E}_2$ holds in the remainders of the proof, remarking that the event $\mathcal{E}_1\cap \mathcal{E}_2$ is independent of what happens during round $T'+1$.

We now distinguish two cases, according to the realization of $T'$.
    First, we assume that $m_{T'-1}=\gamma \in \orange$ and that  $\Expc{W_{T'} \mid m_{T'-1}=\gamma} \not \in \orange$. If $\Expc{W_{T'} \mid m_{T'-1}=\gamma} \not \in \red \cup \orange $, \cref{lem:mt_given_expectation} applied to intervals $I_1 = \blue_1$ and $I_2 = \blue_2 \cup \yellow \cup \green$ implies that $W_{T'} \not \in \red \cup \orange$ and $W_{T'} \leq n/2$ with constant probability. If otherwise $\Expc{W_{T'} \mid m_{T'-1}=\gamma} \in \red$, \cref{lem:red_to_blue} implies $W_{T'} \in \blue_1\cup \blue_2$ with constant probability. Since $m_{T'-1} = \gamma\in \orange$, the two previous statements, together with \cref{lem:preliminaries}(c) implying $U_{T'} = 0$ w.h.p., proving that $m_{T'+1} \not \in \red \cup \orange$ with constant probability, and hence with constant probability $m_{T} \not \in \red \cup \orange$ and $T=T'+1$.

    Next, assume that, for $\gamma \in \orange$, $\Expc{W_{T'}\mid m_{T'-1}=\gamma} \in \orange$ and $m_{T'} \not \in \orange$. In this case,  the event $\mathcal{E}_1 \cap \mathcal{E}_2$ implies $\frac{n}{3} \leq W_{T'} \leq \frac{3n}{4}$ and $U_{T'} = 0$. Since $m_{T'+1} = \min \{W_{T'},n-W_{T'}\}$, this implies that $m_{T'+1} \not \in \red$, concluding that $T=T'+1$ and so $m_T  = m_{T'+1} \not \in \red \cup \orange$ with constant probability.
\end{proof}

%%%%%%%%%%%%%%%%%%%%%%%%%%%%%%%%%%%%%%%%%%%%
% F: Yellow area moved in trunk/yellow.tex %
%%%%%%%%%%%%%%%%%%%%%%%%%%%%%%%%%%%%%%%%%%%%

\subsection{The yellow area}\label{subse:yellow}
The goal of this section is to prove~\cref{lem:out_of_yellow}. The main
difficulty in the analysis of this area is the presence of a fixed point
$\Bar{m} \in \yellow$ for the expectation of the process. Notice that,
according to Fact~\ref{fa:wrong_vs_minor} and Lemma~\ref{lem:no_wrong}, in the
whole \yellow~area it holds that $m_{t+1} = U_t$, w.h.p., for the range of $k$
that we consider. Hence, we here define $\Bar{m}$ as the point such that
$\Expc{U_t \mid m_t = \Bar{m}} = \Bar{m}$. In the following lemma, we show that
such a fixed point is unstable: we first estimate the initial deviation from
$\Bar{m}$ using the variance of the process, then we show that the distance
between $m_t$ and the fixed point grows by a factor $\Omega(\log n)$ in each
round, with a sufficiently large probability. We use this argument to derive an
upper bound on the time it takes the process to leave the \yellow~area.

\begin{lemma}\label{lem:expectation_in_yellow2} 
For any $t_0 \geq 1$, let $T$ be the random variable indicating the first time larger than $t_0$ such that
the process is out of the \yellow~area,
\[
T = \inf \{s > t_0: m_s \not \in \yellow\}\,.
\]
For every $m \in \yellow$, it holds that $\Prob{T \leq \log n \mid m_{t_0} = m}
\geq c$, for a suitable positive constant $c$.
\end{lemma}

\begin{proof}

Let $f(x)=\Expc{U_t \mid m_t=x} =n\left(1-\frac{x}{n}\right)^k +
n\left(\frac{x}{n}\right)^k$.  We have that $f$ is strictly decreasing in
$\yellow$ and, moreover, that $f(x)>x$ if
$x=\frac{n}{k}\log\left(\frac{k}{4\log n}\right)$ and that $f(x)<x$ if
$x=\frac{3n\log n}{k}$. These three facts imply that $f$ admits a unique
fixed point $\Bar{m} \in \yellow$ satisfying $f(\Bar{m})=\Bar{m}$.

Let $\Delta_t = |m_t - \Bar{m}|$.  We prove the following in the remainder: 

\begin{itemize}
\item[(i)] Conditioning on  $\Delta_{t_0} \leq \sqrt{\frac{n}{k}}$, we have
that $\Delta_{t_0+1} \geq \sqrt{\frac{n}{k}}$ with constant probability.
\item[(ii)] For every~$t>t_0$, we have $\Prob{\Delta_{t+1}>\frac{\log
n}{16} \cdot \Delta_t \mid \Delta_t \geq \sqrt{\frac{n}{k}}, t<T}\geq
1-n^{-\frac{1}{33}}$.
\end{itemize}

We prove (i) and (ii) later, while we now show that they imply that $\{T \leq
\log n\}$ with constant probability.  Indeed, if we define $\mathcal{E}_t = \{t
\geq T\} \cup \left\{\Delta_{t+1}>\frac{\log n}{16}\Delta_t \right\}$,
$\mathcal{F} = \left\{\Delta_{t_0+1}\geq \sqrt{\frac{n}{k}}\right\}$ and
$\mathcal{E} =\bigcap_{t=t_0+1}^{t_0+1+\log n}\mathcal{E}_t \bigcap
\mathcal{F}$, we have that, conditioning on $m_{t_0}=m$, $\{T \leq \log n\}
\subseteq \mathcal{E}$. Indeed, if there exists some $t_1$ such that $t_0 \leq
t_1 \leq  t_0 + \log n$ for which,

\begin{equation*}
    \Delta_{t_1}>\max \left\{\frac{3n\log n}{k}-\Bar{m}, \Bar{m}-\frac{n}{k}\log \left(\frac{k}{4\log n}\right)\right\},
\end{equation*}
we have $m_{t_1} \not \in \yellow$. Moreover, we have

\begin{equation}
    \label{eq:out_yellow}
    \Prob{\mathcal{E} \mid m_{t_0} = m} = \prod_{t=t_0+1}^{t_0+1+\log n}\Prob{\mathcal{E}_t \mid \cap_{j=t_0+1}^{t-1}\mathcal{E}_j \cap \mathcal{F}, m_{t_0}=m}\Prob{\mathcal{F} \mid m_{t_0}=m}.
\end{equation}

From (i), we have $\Prob{\mathcal{F} \mid m_{t_0}=m}>c_1$ for a suitable positive
constant $c_1$, while (ii) implies that, for any $t$ such that $t_0+1 \leq t
\leq \log n+t_0+1$,
\[
\Prob{\mathcal{E}_t \mid \cap_{j=t_0+1}^{t-1}\mathcal{E}_j\cap \mathcal{F}, m_{t_0}=m} 
\geq \Prob{\Delta_{t+1} > \tfrac{\log n}{16}\Delta_t \mid \Delta_t \geq \sqrt{\tfrac{n}{k}}, t<T} 
\geq 1-n^{-\frac{1}{33}}\,.
\]
Hence, from \cref{eq:out_yellow} and the inequality below,
$\Prob{\mathcal{E}\mid m_{t_0} = m} \geq \left(1-n^{-\frac{1}{33}}\right)^{\log
n} \cdot c_1 \geq c$, for a suitable positive constant $c$.

\paragraph{Proof of (i)} First, consider the case in which $\Expc{U_{t_0} \mid
m_{t_0} = m}\geq \Bar{m}$ and recall that $m \in \yellow$. From
\cref{lem:preliminaries}(b) we have $\Expc{U_{t_0} \mid m_t = m} \leq
\frac{4n\log n}{k}$ and, from the reverse and standard Chernoff bounds we have
that, with constant probability, $\Bar{m} + \sqrt{\Bar{m}} \leq U_{t_0} \leq
\frac{n}{2}$. Since $\Bar{m} \geq \frac{n}{k}$, this implies 
$m_{t_0+1} \geq \Bar{m} + \sqrt{\frac{n}{k}}$.  Similarly, we can prove that,
when $\Expc{U_{t_0} \mid m_{t_0}=m}\leq \Bar{m}$, we have that $m_{t_0+1}\leq
\Bar{m}-\sqrt{\frac{n}{k}}$ with constant probability. Hence, with constant
probability we have $\Delta_{t_0+1}\geq \sqrt{\frac{n}{k}}$, which
concludes the proof of (i).

\paragraph{Proof of (ii)} Since $\Delta_t \geq \sqrt{\frac{n}{k}}$, we can
write $\Delta_t = \alpha \sqrt{\frac{n}{k}}$, for some $\alpha \geq 1$.  Note
that since $\Bar{m} \in \yellow$ and $k \geq \log n \sqrt{n}$, it satisfies
\begin{equation}\label{eq:char_fx_point}
\Bar{m} \geq \frac{n}{k}\log\left(\frac{k}{3\log n}\right) \geq \frac{n\log n}{3k}.
\end{equation}
First, consider the case when $m_t = \ell \geq \Bar{m}$, i.e., $m_t
=\Bar{m}+\Delta_t$. We have 
\begin{align}
    \notag \Expc{U_t \mid m_t = \ell} \leq \Bar{m} \cdot e^{-\frac{k}{n}\Delta_t} = \Bar{m}e^{-\alpha \sqrt{\frac{k}{n}}} &\leq \Bar{m} \cdot \max \left\{0.3,1-\tfrac{\alpha}{2}\sqrt{\tfrac{k}{n}}\right\}
    \\ &\leq \max \left\{0.3\Bar{m},\Bar{m}-\tfrac{\log n}{8}\alpha \sqrt{\tfrac{n}{k}} \right\}\label{eq:using_fx1} \\
    &\leq \Bar{m}-\frac{\log n}{8}\Delta_t,
    \label{eq:using_fx2}
\end{align}
where the first inequality follows from the fact that $e^{-\frac{1}{x}} \leq
\max \{0.3,1-\frac{1}{2x}\}$ for each $x \geq 0$ and the inequalities
\eqref{eq:using_fx1} and \eqref{eq:using_fx2} follow from
\eqref{eq:char_fx_point}.  Hence, $\Expc{\Bar{m}-U_t \mid m_t = m} \geq
\frac{\log n}{8}\Delta_t$ and so, from using the multiplicative Chernoff bound, we have
\begin{equation}
\Prob{\Bar{m}-U_t \leq \frac{\log n}{16}\Delta_t \mid m_t = \ell} 
\leq e^{-\frac{\log n}{64} \Delta_t} 
\leq n^{-\frac{1}{32}},
\end{equation}
where the last inequality follows since $k \leq 0.5n$. We proved that $U_t \leq
\Bar{m}-\frac{\log n}{16}\Delta_t$ with probability at least
$1-n^{-\frac{1}{32}}$ and, hence, using \cref{lem:preliminaries}(d) we have
$m_{t+1}=U_t$ with probability at least $1-n^{-\frac{1}{33}}$, whence $\Delta_{t+1} \geq \frac{\log n}{16}\Delta_t$ follows.

Now, consider the case in which $m_t = \ell \leq \Bar{m}$, and hence $m_t =
\Bar{m}-\Delta_t$.  In this case, for each $\ell \in \yellow$ we have
\begin{align}
    \Expc{U_t \mid m_t = \ell} \geq ne^{-k\frac{\ell}{n}-k\left(\frac{\ell}{n}\right)^2} \label{eq:using_taylor}  &\geq \Bar{m}e^{\frac{k}{n}\Delta_t}e^{-k\left(\frac{\ell}{n}\right)^2} \\ &\geq \Bar{m} e^{\frac{1}{2}\alpha\sqrt{\frac{k}{n}}}
    \label{eq:using_def_yellow}
    \\ & \geq \Bar{m} +\frac{\Bar{m}\alpha}{2}\sqrt{\frac{k}{n}}
    \label{eq:using_exp}
    \\ &\geq \Bar{m} + \frac{\log n}{6}\alpha \sqrt{\frac{n}{k}}
    \label{eq:using_def_fx2}
\end{align}
where~\eqref{eq:using_taylor} follows from~\cref{claim:exp_taylor},
\eqref{eq:using_def_yellow} follows from the fact that $\ell \in \yellow$ and
so $k\left(\frac{\ell}{n}\right)^2 \leq \frac{1}{2}\alpha \sqrt{\frac{k}{n}}$,
\eqref{eq:using_exp} follows from the fact that $e^{\frac{1}{x}}\geq
1+\frac{1}{x}$ for any $x > 0$. Finally, \eqref{eq:using_def_fx2} follows from
\eqref{eq:char_fx_point}. Next, from the multiplicative Chernoff bound we have

\begin{equation}
\Prob{U_t -\Bar{m} \leq \frac{\log n}{8}\Delta_t \mid m_t = \ell}
\leq e^{-\frac{\log n}{16} \Delta_t}\leq n^{-\frac{1}{8}},
\label{eq:ut_left_yellow}
\end{equation}
where the last inequality follows since $k\leq 0.5n$.  Moreover, from the
definition of \yellow, \cref{lem:preliminaries}(f) implies that $U_t \leq
0.2n$ w.h.p. From a union bound with \eqref{eq:ut_left_yellow}, we have that
$\Bar{m}+\frac{\log n}{8}\Delta_t \leq U_t \leq 0.2n$ with probability at least
$1-n^{-\frac{1}{9}}$. Hence, from \cref{lem:preliminaries}(d) we have 
$m_{t+1}=U_t$ w.h.p., so that $m_{t+1} \geq \Bar{m}+\frac{\log n}{8}\Delta_t$
with probability at least $1-n^{-\frac{1}{10}}$. We thus conclude that, also in
this case,  $\Delta_{t+1}>\frac{\log n}{16} \Delta_t$ with probability at least
$1-n^{-\frac{1}{32}}$.
\end{proof}

We now proceed with the proof of~\cref{lem:out_of_yellow}.

\begin{proof}[Proof of \cref{lem:out_of_yellow}]
In order to prove the lemma, we introduce the stopping time $T' \leq T$, which
can be seen (in some situations) as the time just before $T$:
\[
T' = \inf \{s>t_0: \Expc{U_{s} \mid m_s} \not \in \yellow \text{ or } m_s \not \in \yellow\}\,.
\]
In other words, $T'$ is the first time the expectation of the process, or the
process itself, leaves \yellow. Considering $T'$ instead of $T$ is helpful to
prove that, with constant probability, the minority does not land in the \red~area in round $T=T'+1$.  Consider the following events
\[
A_t = \{m_t \in \yellow\}\cup\{ \Expc{U_t \mid m_t} \in \yellow\},
\]
\[
B_t = \{U_t \in \blue_2 \cup \yellow \cup \green\}\cap \{W_t = 0\},
\]
and
\[
\mathcal{E}_t = (A_t)^C\cup B_t \cup \{t \geq T'\}, 
\qquad \mathcal{E}_1 = \bigcap_{t=t_0}^{t_0+\log n} \mathcal{E}_t 
\quad \text{and} \quad 
\mathcal{E}_2 = \{T' \leq \log n\}. 
\]
We have  
\begin{align*}
\Prob{\mathcal{E}_1^C \mid m_{t_0}=m} 
& \leq \sum_{t=t_0}^{t_0+\log n}
\Prob{\mathcal{E}_t^C \mid m_{t_0}= m} \\ 
& \leq \sum_{t=t_0}^{t_0+\log n}\Prob{(B_t)^C \mid A_t,t<T',m_{t_0}=m} 
\\ & \leq \frac{2\log n}{n^2},
\end{align*}
where the last inequality follows from \cref{lem:mt_given_expectation} taking
the interval $I=\yellow$ and noticing that $2I \subseteq \blue_2 \cup \yellow
\cup \green$, from \cref{lem:preliminaries}(d) and a union bound.

From~\cref{lem:expectation_in_yellow2} and noticing that $T'\leq T$, we have
that there exists a constant $c_2>0$ such that $\Prob{\mathcal{E}_2 \mid
m_{t_0}=m} \geq c_2$. Hence, we have  $\Prob{\mathcal{E}_1 \cap
\mathcal{E}_2 \mid m_{t_0}=m} \geq c_2-\frac{1}{n} >0$, whenever $n$ is sufficiently
large. In the rest of the proof, we assume that the event $\mathcal{E}_1 \cap
\mathcal{E}_2$ holds. We remark that the event $\mathcal{E}_1 \cap
\mathcal{E}_2$ is independent of what happens in round $T'+1$.

We now consider two cases, according to the realization of $T'$. First, we assume
that $m_{T'-1} = \gamma \in \orange$ and that $\Expc{U_{T'} \mid m_{T-1}=\gamma
}\not \in \yellow$. In this case, we note that, since $\gamma \in \yellow$,
from \cref{lem:preliminaries}(c) and the strong Markov property, we have
\[
\frac{1}{n^2}\leq \Expc{U_{T'} \mid m_{T'-1} =\gamma} \leq  \frac{3n\log n}{k}\,.
\]
Since we assumed $\Expc{U_{T'} \mid m_{T'-1} =\gamma} \not \in \yellow$,
it follows that $\Expc{U_{T'} \mid m_{T'-1}} \in \blue_1 \cup \red \cup \blue_2
\cup \green$.  We first consider the case in which $\Expc{U_{T'} \mid m_{T'-1}=
\gamma} \in \red$. From \cref{lem:red_to_blue}, we have $U_{T'} \in
\blue_1\cup \blue_2$ with constant probability, while
\cref{lem:preliminaries}(d) implies that $m_{T'+1} \in \blue_1\cup \blue_2$,
thus proving that, in this case, $T=T+1$ and that $m_{T'+1} = m_{T} \in \blue_1 \cup
\blue_2$ with constant probability.  Now assume that $\Expc{U_{T'} \mid
m_{T'-1}} \in \blue_1 \cup \blue_2 \cup \green$. We can apply
\cref{lem:mt_given_expectation} to the intervals $I_1 = \blue_1$, $I_2=
\blue_2$ and $I_3 = \green$, obtaining that, with constant probability, $U_{T'}
\in \blue_1 \cup \blue_2 \cup \green$. Considering also
\cref{lem:preliminaries}(d), we have that $T=T'+1$ and that
$U_{T'}=m_{T'+1}= m_T \in \blue_1 \cup \blue_2\cup \green$.

Suppose next that $T'$ is such that $m_{T'-1}= \gamma \in \yellow$, $\Expc{U_{T}
\mid m_{T'-1}=\gamma}\in \yellow$ but $m_{T'} \not \in \yellow$. In this case,
the event $\mathcal{E}_1 \cap \mathcal{E}_2$ implies that $T=T'+1$ and
$U_{T'}=m_{T'+1}=m_T \in \blue_1 \cup \green$.
\end{proof}

%%%%%%%%%%%%%%%%%%%%%%%%%%%%%%%%%%%%%%%%%%%%%%%%%%%%%%%%
% F: End of old Yellow area, moved in trunk/yellow.tex %
%%%%%%%%%%%%%%%%%%%%%%%%%%%%%%%%%%%%%%%%%%%%%%%%%%%%%%%%

\subsection{The blue area}

In this section, we prove \cref{lem:out_of_blue}.

\begin{proof}[Proof of \cref{lem:out_of_blue}]
We first assume that $m \in \blue_1$. From \cref{lem:preliminaries}(b) and since $m \leq \frac{n\log 2}{k}-\sqrt{\frac{n\log 2}{k}}$, we obtain the following bound:
\begin{align}
    \Expc{U_t \mid m_t = m} \geq e^{-k\frac{m}{n}-k\left(\frac{m}{n}\right)^2} 
    &\geq \frac{n}{2}e^{\sqrt{\frac{k\log 2}{n}}}e^{-k\left(\frac{\log 2}{k}-\sqrt{\frac{\log 2}{kn}}\right)^2} 
    \label{eq:using_me} \\ \notag &\geq  \frac{n}{2}e^{\frac{1}{2}\sqrt{\frac{k\log 2}{n}}} \\ &\geq \frac{n}{2} +\frac{n}{4}\sqrt{\frac{k\log 2}{n}},
    \label{eq:exp_using}
\end{align}
where \eqref{eq:using_me} follows from \cref{claim:exp_taylor}, and \eqref{eq:exp_using} follows from the fact that $e^{\frac{1}{x}}\geq 1+\frac{1}{x}$ for any $x>0$. 
On the other side, from \cref{lem:preliminaries}(b), we have that, since $m_t \geq 1$, $\Expc{U_t \mid m_t = m} \leq ne^{-\frac{k}{n}} \leq n-\frac{k}{3} \leq n-\frac{n\log n}{k}$, since $k \geq 185\sqrt{n \log n}$. Hence, we proved that $\Expc{n-U_t \mid m_t = m} \in \green$ and, from \cref{lem:mt_given_expectation} applied to $I=\green$ and \cref{lem:preliminaries}(d), we have that $m_{t+1} = n-U_t \in \green$ with constant probability.

We now assume that~$m \in \blue_2$, and hence $  \frac{n\log 2}{k} +\sqrt{\frac{n\log 2}{k}} \leq m \leq \frac{n}{k} \log \pa{\frac{k}{4 \log n}}.$
In this case, from \cref{lem:preliminaries}(b), we have that
\begin{align*}
    \Expc{U_t \mid m_t = m} \leq ne^{-\frac{k}{n}m}&\leq \frac{n}{2}e^{-\sqrt{\frac{k\log 2}{n}}} +\frac{1}{n^2}\\ & \leq \frac{n}{2}-\frac{n}{5}\sqrt{\frac{k\log 2}{n}},
\end{align*}
where the last inequality follows since $e^{-1/x}\leq 1-\frac{1}{2x}$ for any $x \geq 1$. On the other side, 
\begin{align*}
    \Expc{U_t \mid m_t = m} \geq ne^{-k\frac{m}{n}-k\left(\frac{m}{n}\right)^2} &\geq \frac{4n\log n}{k} e^{-\log^2\left(\frac{k}{4\log n}\right)\frac{1}{k}} \\& \geq \frac{3n\log n}{k},
\end{align*}
Hence, $\Expc{U_t \mid m_t = m} \in \green$ and, from \cref{lem:mt_given_expectation} applied to the interval $I= \green$ and \cref{lem:preliminaries}(d), we have that $m_{t+1} = U_t \in \green$ with constant probability.
\end{proof}

\subsection{The green area}

\begin{proof} [Proof of \cref{lem:out_of_green}]
By definition of $\green$, we can apply \cref{lem:preliminaries}(c) and \cref{lem:preliminaries}(d), obtaining that $m_{t+1} = 0$ w.h.p.
We notice that, in this case, $\mathcal{O}_{t+1}=1-\mathcal{O}_t$, since every node adopts the ex-minority opinion.
\end{proof}

\subsection{The Bit-Dissemination problem}\label{subse:bit_dissemination}

\begin{proof} [Proof of \cref{lem:wrong_consensus}]
    By assumption, $m_t =m= 1$ and the opinion of the source agent is $\ell = \minop_t$. Since, from the fact that $k \leq 0.5n$, we have that $m =1 \leq \frac{n}{3k}$, we have from \cref{lem:preliminaries}(e)  $\mathcal{O}_{t+1}=\mathcal{O}_t = \ell$ w.h.p.~and, from \cref{lem:out_of_blue}, that $m_{t+1} \in \green$ with constant probability. Finally, from \cref{lem:out_of_green}, we have that $m_{t+2}= 0$ and that $\mathcal{O}_t = 1-\ell$, and therefore the opinion of the majority is the same opinion of the source.
    \end{proof}

\section{\texorpdfstring{$k$-Minority}{k-Minority} in the Sequential Communication Model} \label{sec:lowerbound_seq}
In this section, we analyze the \minority\ process on the uniform-random sequential model, proving \cref{thm:lb_sequential}.

\subsection{Proof of \texorpdfstring{ \cref{thm:lb_sequential}}{Theorem 4}}
In this section, we will denote by $X_t$ the number of agents with opinion $1$ in round $t$.
In addition, we will use the following notations for birth-death chains on~$\{0,\ldots,n\}$:
\begin{align*}
    p_i &= \Prob{X_{t+1} = i+1 \mid X_t = i}, \\
    q_i &= \Prob{X_{t+1} = i-1 \mid X_t = i},  \\
    r_i &= \Prob{X_{t+1} = i \mid X_t = i}, \text{ and } \\
    \tau_{i,j}&= \Expc{ \inf \left\{ t \in \mathbb{N}, X_t = j \right\} \mid X_0 = i}.
\end{align*}
We start by recalling a classical lower bound on the expected time needed for a birth-death chain to travel from state~$0$ to $n$.
\begin{lemma} \label{lem:lower_bound_hit}
    Consider any birth-death chain on~$\{0,\ldots,n\}$. For~$1 \leq i \leq j \leq n$, let $a_i = q_i/p_{i-1}$ and $a(i:j) = \prod_{k=i}^j a_k$. Then, $\tau_{0,n} \geq \sum_{1 \leq i < j \leq n} a(i:j)$.
\end{lemma}
\begin{proof} [Proof (for the sake of completeness)]
    Let~$w_0 = 1$ and for $i \in \{1,\ldots,n\}$, let $w_i = 1/a(1:i)$. The following result is well-known (see, e.g., Eq.~(2.13) in \cite{levin2017markov}). For every~$\ell \in \{1,\ldots,n\}$,
    \begin{equation*} 
        \tau_{\ell-1,\ell} = \frac{1}{q_\ell w_\ell} \sum_{i=0}^{\ell-1} w_i.
    \end{equation*}
    Thus,
    \begin{equation*}
    	\tau_{\ell-1,\ell} = \frac{1}{q_\ell} \sum_{i=0}^{\ell-1} \frac{a(1:\ell)}{a(1:i)} = \frac{1}{q_\ell} \sum_{i=1}^{\ell} a(i:\ell) \geq \sum_{i=1}^{\ell} a(i:\ell).
    \end{equation*}
    Eventually, we can write
    \begin{equation*}
    	\tau_{0,n} = \sum_{\ell = 1}^n \tau_{\ell-1,\ell} \geq \sum_{1 \leq i < j \leq n} a(i:j),
    \end{equation*}
    which concludes the proof of \cref{lem:lower_bound_hit}.
\end{proof}

Next, we identify some regions of the configuration space where $\minority$ has a very strong drift toward the balanced configuration $X_t = n/2$.
\begin{lemma} \label{lem:proba_in_boundary_region}
    Let~$\alpha,\beta$ such that $0<\alpha<\beta<1/2$. There exists a constant~$c=c(\alpha,\beta)>0$, such that
    for every $n$ large enough,
    for every~$i \in [\alpha n, \beta n]$,
    
    \begin{equation*}
        \Prob{X_{t+1} = i+1 | X_t = i} \geq  \frac{1-\beta}{2} \, 
        \mbox{ and } \ \Prob{X_{t+1} = n- i-1 | X_t = n-i} \geq \frac{1-\beta}{2};
 \end{equation*} 
    
    and 
    \begin{equation*}
        \Prob{X_{t+1} = i-1 | X_t = i}  \leq \exp \pa{-c \, k} \, \mbox{ and } \ \Prob{X_{t+1} = n-i+1 | X_t = n-i} \leq \exp \pa{-c \, k}.
    \end{equation*}
\end{lemma}
\begin{proof}
    Let $i \in [\alpha n, \beta n]$.
    Let~$A_t$ the event: ``In round~$t$, the activated agent sees a unanimity of~$0$''.
    Since $i \geq \alpha n$,
    \begin{equation*}
        \Prob{A_t \mid X_t = i} \leq  (1-\alpha)^k.
    \end{equation*}
    Let~$B_t$ the event: ``In round~$t$, the activated agent sees a minority of~$0$''.
    Let $Y$ be a random variable following a binomial distribution with parameters~$(k,\beta)$.
    Since $i \leq \beta n$, by the additive Chernoff bound (\cref{thm:additive_chernoff}),
    \begin{align*} 
        \Prob{B_t \mid X_t = i} &\leq \Prob{Y \geq k/2} \\
        &= \Prob{Y \geq \beta k + k(1/2-\beta)} \\
        &\leq \exp \pa{ \frac{-2k^2(1/2-\beta)^2}{k}} \\
        &= \exp \pa{-2k\pa{1/2-\beta}^2}.
    \end{align*}
    Let~$C_t$ the event: ``In round~$t$, the activated agent adopts opinion~$0$''.
    By construction of the Minority protocol, $C_t \subseteq A_t \cup B_t$.
    Taking the union bound, we get
    \begin{equation*}
        \Prob{C_t \mid X_t = i} \leq  (1-\alpha)^k + \exp \pa{-2k\pa{1/2-\beta}^2}.
    \end{equation*}
    By setting
    \begin{equation*}
        c = \min \left\{ \log\pa{\frac{1}{1-\alpha}} , 2\pa{\frac{1}{2}-\beta}^2 \right\} > 0,
    \end{equation*}
    we obtain
    \begin{equation*}
        \Prob{C_t \mid X_t = i} \leq 2 \exp \pa{-c \, k}.
    \end{equation*}
    Therefore, we have
    \begin{equation*}
        \Prob{X_{t+1} = i-1 \mid X_t = i} = \frac{i}{n} \cdot \Prob{C_t \mid X_t = i} \leq \beta \cdot 2 \exp \pa{-c \, k} \leq \exp \pa{-c \, k}.
    \end{equation*}
    Moreover, for $k$ large enough,
    \begin{equation*}
        \Prob{X_{t+1} = i+1 \mid X_t = i} 
        = \frac{n-i}{n} \cdot (1-\Prob{C_t \mid X_t = i})
        \geq (1-\beta) \pa{1-2 \exp \pa{-c \, k}} \geq \frac{1-\beta}{2}.
    \end{equation*}
    By symmetry of $\minority$, we obtain the same bounds for the case $X_t = n-i$,
    which concludes the proof of \cref{lem:proba_in_boundary_region}.
\end{proof}

Eventually, we can combine the above results into a lower bound on the convergence time of $\minority$ in the sequential setting.
\begin{proof} [Proof of \cref{thm:lb_sequential}]
    Let $m=n/6$.
    Let $Z = (Z_s)$ be a birth-death chain on $\{0,\ldots,m\}$, defined as follows: for every $\delta \in \{-1,0,1\}$, for every~$i \in \{1,\ldots,m-1\}$,
    \begin{equation*}
        \Prob{Z_{s+1} = i+\delta \mid Z_s = i} =
        \Prob{X_{t+1} = \frac{n}{2}+m+i+\delta \mid X_t = \frac{n}{2}+m+i}.
    \end{equation*}
    Note that, by symmetry, this is also equal to
    $\Prob{X_{t+1} = \frac{n}{2}-m-i-\delta \mid X_t = \frac{n}{2}-m-i}$.
    Intuitively, making one step {\em forward} as $Z_s$, corresponds to making one step {\em away from $n/2$} as $X_t$.
    In particular, with this construction, we have that
    \begin{equation*}
        \Expc{\inf \{t \in \mathbb{N}, |X_t-n/2| \geq 2m \} \mid X_0 = n/2 } \geq \Expc{\inf \{s \in \mathbb{N}, Z_s = m \} \mid Z_0 = 0}.
    \end{equation*}    
    By \cref{lem:proba_in_boundary_region} applied with $\alpha = 1/6$, $\beta = 2/6$, there exists a constant $c>0$ such that for every~$i \in \{1,\ldots,m\}$, $Z$ satisfies $p_i \leq \exp \pa{-c \, k}$ and $q_i \geq 1/3$.
    As a consequence, for every~$i \in \{1,\ldots,m\}$, $q_i/p_{i-1} \geq \exp \pa{c \, k}/3$.
    Finally, by \cref{lem:lower_bound_hit},
    \begin{equation*}
        \Expc{\inf \{s \in \mathbb{N}, Z_s = m \} \mid Z_0 = 0} \geq \sum_{1 \leq i < j \leq n} a(i:j) \geq \pa{ \frac{\exp \pa{c \, k}}{3} }^{n-1} = \exp \pa{\Omega(n \, k)},
    \end{equation*}
    which concludes the proof of \cref{thm:lb_sequential}.
\end{proof}

\section{Open Questions} \label{sec::open}

We see two main questions that deserve future work.
The first question is more technical and concerns the role of parameter $k$, i.e. the size of the sample used by nodes to apply the \minority\ rule. Despite our proofs of \cref{thm:main} and \cref{cor:main} require $k = \tilde\Omega(\sqrt{n})$ to achieve $O(\mathrm{polylog} n)$ convergence time, preliminary experimental results suggest that poly$\log n$ convergence time is possible even when   $k = \Omega(\mathrm{polylog} n)$, and it would be interesting to establish such bounds rigorously or prove that actual bounds might actually be worse. So far, we are able to prove that if $k=O(1)$ then any passive-communication stateless dynamics in the $k$-$\pull$ model requires $n^{\Omega(1)}$ rounds to solve the bit-dissemination problem.

The second, more general question concerns the gap between the synchronous parallel and the asynchronous sequential models  in \textit{consensus} dynamics, or possibly in dynamics for other fundamental distributed  tasks such as \textit{broadcast}. A $\tilde O(n)$ gap between the two models is so common that, in order to analyze a parallel dynamics with (presumed) convergence time poly$\log n$, it is common to first analyze its sequential implementation, to establish an $O(n \, {\rm polylog } \, n)$
bound on converge time, and then to adapt the analysis to the parallel setting. 
Our results show a natural, simple dynamics for which the usual linear gap does not hold,  showing that it is not possible to establish a black-box result, whereby results for the  asynchronous sequential setting immediately carry over to the parallel synchronous case.
We believe it would be very interesting to derive necessary and/or sufficient conditions under which a linear gap is guaranteed to exist.

\paragraph{Acknowledgments.} The authors wish to thank Amos Korman and Emanuele Natale for very helpful discussions on the topic.

\bibliographystyle{plain}
\bibliography{main}

\clearpage

\appendix

\section{Tools}

\begin{theorem}[Chernoff's Inequality]
\label{thm:additive_chernoff}
Let $X=\sum_{i=1}^n X_i$, where $X_i$ with $i \in [n]$ are independently
distributed in $[0,1]$. Let $\mu=\Expc{X}$ and $\mu_- \leq \mu \leq \mu_+$.
Then:
\begin{itemize}
    \item for  every  $t>0$
    \[
    \Prob{X>\mu_+ +t}\leq  e^{-2t^2/n}  \quad \text{and} 
 \quad \Prob{X<\mu_- -t}\leq e^{-2t^2/n} ;
    \]
    \item for $\varepsilon>0$
    \[
    \Prob{X>(1+\varepsilon)\mu_+}\leq e^{-\frac{\epsilon^2}{2}\mu_+} 
    \ \text{ and } \
    \Prob{X<(1-\epsilon)\mu_-}\leq e^{-\frac{\epsilon^2}{2}\mu_-} . 
    \]
\end{itemize}

\end{theorem}

\begin{lemma}[Reverse Chernoff Bound] \label{lem:add_reverse_chernoff}
    Let $X_1,\dots, X_n$ be i.i.d. Bernoulli random variables. Let $X = \sum_{i=1}^n X_i$ and let $X_i$ such that $\Expc{X_i} = p$, and $\mu= \Expc{X} = np$. Then, if $p \leq 1/4$, for any $t \geq 0$
    \[\Prob{X \geq \mu + t} \geq \frac{1}{4}e^{-\frac{2t^2}{\mu}}.\]
    If $p \leq 1/2$, then for any $0 \leq t \leq n(1-2p)$,
    \[\Prob{X>t+\mu} \geq \frac{1}{4}e^{-\frac{2t^2}{\mu}}.\]
    Furthermore, for any $\delta \in (0,\frac{1}{2}]$ such that $\delta^2 \mu \geq 3$,
    	\begin{align*}
		\Prob{X\ge (1+\delta)\mu}\ge e^{-9 \delta^2 \mu} \quad \text{and}\quad 
		\Prob{X \le (1-\delta)\mu}\ge e^{-9\delta^2\mu}.
		\end{align*}
    \label{lem:chernoff-reverse2}
    
\end{lemma}

\begin{claim} \label{claim:exp_taylor}
    If $x/n\leq 0.6$, then for $n$ large enough, the following inequalities hold:
    \begin{equation*}
        \exp \pa{-2k \, \frac{x}{n}} \leq \exp \pa{-k \, \frac{x}{n} -k \, \pa{\frac{x}{n}}^2 } \leq \left(1-\frac{x}{n}\right)^k \leq \exp \pa{-k \, \frac{x}{n}}.
    \end{equation*}
\end{claim}

\begin{claim}[Central Binomial Coefficient]
\label{claim:centralbinomial}
For every large enough $n$, it holds
\[\frac{1}{2}\cdot \frac{4^n}{\sqrt{\pi n}} \leq \binom{2n}{n}\leq \frac{4^n}{\sqrt{\pi n}}.\]
\end{claim}

\highlight{
\section{Missing Proofs} \label{app:preliminaries_proof}
}

\begin{proof}[Proof of \Cref{lem:preliminaries}]
    We prove each item separately.
           
    \paragraph{Proof of (a).}
     Let $X \in \{0,\dots, k\}$ the number of samples corresponding to the minority opinion that node $i$ \highlight{receives} in round $t$. Since $m_t = m$, each sample corresponds to the minority opinion with probability $m/n$. Therefore, we have $\Expc{X \mid m_t=m} = k \cdot \frac{m}{n} = \frac{k}{2} - k\left(\frac{1}{2}-\frac{m}{n}\right)$.
    Hence, by Chernoff Bound and the Reverse Chernoff Bound, we have 
    \[ \frac{1}{4}e^{-4k\left(\frac{1}{2}-\frac{m}{n}\right)^2} \leq \Prob{X \geq \left\lceil \frac{k}{2}\right\rceil \mid m_t=m} \leq e^{-2k\left(\frac{1}{2}-\frac{m}{n}\right)^2}.\]
    Now, given $m_t = m$, since $W_t$ is the sum of $n$ i.i.d. Bernoulli random variables taking value $1$ when $\{X \geq \left\lceil k/2\right \rceil\}$, we have that (a) holds.

    \paragraph{Proof of (b).}  The generic agent~$i$ samples 
    only the majority opinion with probability  $(1-m/n)^k$. Summing over all agents and taking the expectation concludes the proof of (b).

    \paragraph{Proof of (c).}  If $m_t = m \geq \frac{3n\log n}{k}$, (b) implies that $\Expc{U_t \mid m_t=m} \leq \frac{2}{n^2}$. By Markov's inequality, $\Prob{U_t = 0 \mid m_t = m} \geq 1-\frac{2}{n^2}$, which concludes the proof of (c).

    \paragraph{Proof of (d).} Since $m_t = m \leq \frac{n}{2}-n\sqrt{\frac{1.5\log n}{k}}$, we have $\left(\frac{1}{2}-\frac{m}{n}\right)^2 \geq \frac{1.5\log n}{k}$, whence (a) implies $\Expc{W_t \mid m_t=m}\leq \frac{1}{n^2}$. By Markov's inequality, $\Prob{W_t = 0 \mid m_t = m} \geq 1-\frac{1}{n^2}$, which concludes the proof of (d).

    \paragraph{Proof of (e).} If $m_t=m \leq \frac{n}{3k}$, (b) yields $\Expc{U_t \mid m_t=m} \geq n \pa{1-\frac{m}{n}}^k \geq ne^{-\frac{1.1k}{n}m} \geq 0.69n$. 
    Then, using Chernoff's bound we have $
        \Prob{U_t \leq 0.6n \mid m_t=m}\leq e^{-2\cdot (0.09)^2 n} \leq \frac{1}{n^2}$. Moreover, since $m \leq \frac{n}{3k}$, (d) gives $W_t=0$  w.h.p., with no agent seeing a unanimity that corresponds to $\minop_t$.
    Thus, by definition of the protocol, any agent seeing unanimity adopts opinion~$1-\minop_t$.
    Since $U_t \geq 0.6n$, this implies that a minority of agents adopt opinion $\minop_t$, hence $\minop_{t+1}=\minop_t$, which concludes the proof of (e).

    \paragraph{Proof of (f).}  
    If $m_t =m$ and $\frac{2n}{k}\leq m \leq \frac{n}{2}-n\sqrt{\frac{1.5\log n}{k}}$, (b) gives $
        \Expc{U_t \mid m_t=m} \leq ne^{-\frac{k}{n}m}\leq 0.15n$.
    Hence, using Chernoff's bound we have
    $\Prob{U_t \geq 0.2n \mid m_t=m}\leq e^{-2(0.05)^2 n} \leq \frac{1}{n^2}$.
    Moreover, from (d) we have $W_t=0$ w.h.p. The latter event implies that agents adopt the majority opinion if and only if they sample unanimity. Hence, conditioning on the fact that $U_t \leq 0.2n$, we have that $\minop_{t+1}=1-\minop_t$, which concludes the proof of (f).
\end{proof}

\end{document}